\documentclass[11pt]{article}

\usepackage[margin=1in]{geometry}
\usepackage{amsmath,amssymb,amsthm}
\usepackage{graphicx}
\usepackage{xcolor}
\usepackage{placeins}
\usepackage{authblk}
\usepackage{tikz}
\usetikzlibrary{arrows.meta}
\usepackage[singlelinecheck=false, font=small]{caption}

\tikzset{every picture/.style={line width=0.75pt}}

\newtheorem{theorem}{Theorem}[section]
\newtheorem{lemma}[theorem]{Lemma}
\newtheorem{observation}[theorem]{Observation}
\newtheorem{corollary}[theorem]{Corollary}

\theoremstyle{definition}
\newtheorem{definition}[theorem]{Definition}

\newcommand{\N}{\mathbb{N}}
\newcommand{\Q}{\mathbb{Q}}

\newcommand{\Z}{\mathbb{Z}}

\DeclareMathOperator{\adim}{adim}
\DeclareMathOperator{\aDim}{aDim}
\DeclareMathOperator{\odim}{odim}
\DeclareMathOperator{\oDim}{oDim}
\DeclareMathOperator{\strong}{{str}}
\DeclareMathOperator{\mh}{MH}

\newcommand{\odimFS}[1]{\odim_{\textup{FS}}^{#1}}
\newcommand{\oDimFS}[1]{\oDim_{\textup{FS}}^{#1}}
\newcommand{\adimFS}[1]{\adim_{\textup{FS}}^{#1}}
\newcommand{\aDimFS}[1]{\aDim_{\textup{FS}}^{#1}}

\definecolor{gold}{RGB}{245,166,35}

\title{Adaptive Multi-Head Finite-State Gamblers}

\author[1]{Julianne Cruz}
\author[1]{Sho Glashausser}
\author[2]{Xiaoyuan Li}
\author[1]{Neil Lutz}

\affil[1]{Swarthmore College}
\affil[2]{Iowa State University}

\date{}

\begin{document}

\maketitle

\begin{abstract}
    Multi-head finite-state dimensions and predimensions quantify the predictability of a sequence by a gambler with trailing heads acting as ``probes to the past.''  These additional heads allow the gambler to exploit patterns that are simple but non-local, such as in a sequence $S$ with $S[n]=S[2n]$ for all $n$. In the original definitions of Huang, Li, Lutz, and Lutz (2025), the head movements were required to be \emph{oblivious} (i.e., data-independent). Here, we introduce a model in which head movements are \emph{adaptive} (i.e., data-dependent) and compare it to the oblivious model. We establish that for each $h\geq 2$, adaptivity enhances the predictive power of $h$-head finite-state gamblers, in the sense that there are sequences whose oblivious $h$-head finite-state predimensions strictly exceed their adaptive $h$-head finite-state predimensions. We further prove that adaptive finite-state predimensions admit a strict hierarchy as the number of heads increases, and in fact that for all $h\geq 1$ there is a sequence whose adaptive $(h+1)$-head finite-state predimension is strictly less than its adaptive $h$-head predimension.
\end{abstract}

\section{Introduction}

    This paper initiates the study of adaptive multi-head finite-state dimensions and predimensions, showing that the predictive power of multi-head finite-state gamblers depends both on the adaptivity of their head movements and on the number of heads.
    
    Finite-state dimensions were introduced by Dai, Lathrop, Lutz, and Mayordomo~\cite{DLLM04} and Athreya, Hitchcock, Lutz, and Mayordomo~\cite{AHLM07}. As with other effective dimensions~\cite{DISS,DCC}, they measure the unpredictability of an infinite sequence $S$ by considering the rates at which a gambler placing bets on successive symbols of $S$ can amass capital. For finite-state dimensions, the gambler's betting function must be implementable by a finite-state machine. There is a large body of literature on these dimensions and their applications, especially to topics related to Borel's theory of normal numbers; see~\cite{KozShe2021} for some of these applications.
    
    Due to the strict computational constraint on the betting function, finite-state gamblers are only able to exploit highly localized patterns in the sequence. For example, given a sequence $S$ where $S[n]=S[2n]$ for all $n\in\N$, an unconstrained gambler can quickly amass capital by remembering the whole sequence and placing ``all-in'' bets on the even-numbered symbols, but a finite-state gambler has no way to take advantage of this simple pattern more than a constant number of times.
    
    Inspired by the work of Becher, Carton, and Heiber~\cite{BCH2018} on \emph{finite-state independence}, Huang, Li, Lutz, and Lutz~\cite{mhfsd} recently introduced \emph{multi-head finite-state gamblers} that are still strictly constrained but are well-suited to this type of non-local pattern. These gamblers are still required to implement their betting functions with a finite-state machine, but they operate under a different data access model: rather than accessing transitory symbols only once, this model permits the gambler to have additional read heads acting as ``probes to the past,'' accessing symbols from earlier in the sequence. Using this model,~\cite{mhfsd} defined, for each positive integer $h$, the \emph{$h$-head finite-state predimension} of a sequence, where the original definition of finite-state dimension~\cite{DLLM04} is the $h=1$ case.
    
    The term \emph{predimension} was used in~\cite{mhfsd} because these quantities lack a standard stability property expected of dimension notions; that work then defined the stable quantity \emph{multi-head finite state dimension}, as well as versions of these notions that are \emph{strong} in the sense of~\cite{AHLM07}. Lutz~\cite{mhfsc} subsequently showed that these predimensions and dimensions admit equivalent characterizations in terms of lossless compression by multi-head finite-state transducers.
    
    The main result of~\cite{mhfsd} was a hierarchy theorem, showing that for all $h\geq 1$, there is a sequence whose $(h+1)$-head finite-state predimension is strictly less than its $h$-head finite-state predimension. This is reminiscent of the classic hierarchy theorem by Yao and Rivest~\cite{YR78} showing that ``$k+1$ heads are better than $k$'' in the context of language recognition by multi-head automata, and the study of multi-head finite-state gamblers is broadly analogous to the long line of work on multi-head automata.
    
    There is an important distinction, however, in the rules governing head movement. In the definitions of~\cite{mhfsd}, the read heads are required to move in an \emph{oblivious} (or \emph{data-independent}) way, meaning they move right according to a pre-determined schedule that does not depend on the input string. By contrast, most work on multi-head automata---including the hierarchy theorem of Yao and Rivest~\cite{YR78}---concerns \emph{adaptive} (or \emph{data-dependent}) head movements~\cite{HKM09}, although oblivious variants have also been studied (e.g.,~\cite{DKP20}). Hence, in this work, we define and study adaptive multi-head finite-state gamblers, predimensions, and dimensions.
    
    Our first main result is a separation theorem, showing that for all $h\geq 2$, there is a sequence whose adaptive $h$-head finite-state predimension is strictly less than its oblivious $h$-head finite-state predimension. We prove this by constructing a family of sequences in which the ``correct'' relative position for one of the read heads changes over time in a way described by the sequence itself. An adaptive gambler can adjust that head's speed accordingly to keep it in the correct position, but an oblivious gambler can make no such adjustment. We prove an upper bound on the adaptive $h$-head finite-state predimension of this sequence by constructing a suitable gambler, and we use a Kolmogorov complexity argument to derive a lower bound on the sequence's oblivious $h$-head finite-state predimension.
    
    Our second main result is a hierarchy theorem, directly analogous to the hierarchy theorem of Yao and Rivest~\cite{YR78}: we show that for all $h\geq 1$, there is a sequence whose adaptive $(h+1)$-head finite-state predimension is strictly less than its adaptive $h$-head finite-state predimension. We prove this adaptive hierarchy using the same family of sequences used in~\cite{mhfsd} to prove an oblivious hierarchy, and our proof tracks their Kolmogorov complexity arguments closely. It is surprising, in light of our separation theorem and the additional flexibility granted by adaptivity, that the same techniques can be so straightforwardly applied in this setting.
    
    The rest of the paper is organized as follows. In Section~\ref{sec:gamblers}, we define adaptive and oblivious multi-head finite-state gamblers. In Section~\ref{sec:prelim}, we present the algorithmic information ideas used in our later definitions and proofs, including Kolmogorov complexity and gales, and we give a general lemma for translating Kolmogorov complexity bounds (which concern compressibility) into gale bounds (which concern unpredictability) in the context of multi-head finite-state gamblers. In Section~\ref{sec:dims}, we use gales to define adaptive and oblivious multi-head finite-state predimensions and dimensions. In Section~\ref{sec:seq}, we define the family of sequences used in our separation theorem, and then we prove adaptive upper bounds and oblivious lower bounds on this sequence family in Sections~\ref{sec:ub} and~\ref{sec:lb}, respectively. In Section~\ref{sec:hier}, we prove our strict hierarchy theorem, and we conclude with several directions for future work in Section~\ref{sec:conc}.

\section{Adaptive Multi-Head Finite-State Gamblers}\label{sec:gamblers}

    \subsection{String and Sequence Notation}

        Throughout this paper, an alphabet is a finite set $\Sigma$ such that $|\Sigma|\geq 2$. Sequences and strings are indexed starting from 0, and for each sequence $S$ and $m\leq n$, we write $S[m..n]$ for the $(n-m+1)$-symbol string $S[m]S[m+1]\ldots S[n]$. The empty string is denoted $\lambda$, and $S[m..n]=\lambda$ when $m>n$. We denote the length of a string $x$ by $|x|$ and write $\sqsubseteq$ for the prefix relation.
        
        We will often consider finite subsequences of infinite sequences, and it will be convenient to retain the position of each symbol in the subsequence, for which we use placeholder symbols: Given a fixed ``horizon'' $n\in\N$, a sequence $S$ over some alphabet $\Sigma$, and a set $A\subseteq[0,n]$, we let $S[A]$ denote the $(n+1)$-symbol string such that
        \[S[A][i]=\begin{cases}S[i]&\text{if }i\in A\\0&\text{if }i\not\in A.\end{cases}\]
        Here 0 indicates the lexicographically first symbol in $\Sigma$; in particular, the placeholder symbol will be $0^L$ when the alphabet is $\{0,1\}^L$.  Note that $S[A]=S[A\cap\N]$ for all $A$ and that $S[[m,n]]$ is not the $(n-m+1)$-symbol string $S[m..n]$; instead, it is the $(n+1)$-symbol string consisting of $m$ 0s followed by $S[m..n]$.

    \subsection{Adaptive Multi-Head Finite-State Gamblers}
        We now define adaptive multi-head finite-state gamblers as a generalization of the multihead finite-state gamblers defined in~\cite{mhfsd}, which we reclassify below as \emph{oblivious} multihead finite-state gamblers. Note that while \emph{adaptive} often means \emph{not oblivious} in informal language, our formal definitions treat oblivious gamblers as a special case of adaptive gamblers.
        \begin{definition}
            Let $h$ be a positive integer. An \emph{adaptive $h$-head finite-state gambler} (\emph{adaptive $h$-FSG}) is a 6-tuple $G=(Q,\Sigma,\delta,\beta,q_0,c_0)$, where
            \begin{itemize}
                \item $Q$ is a finite, non-empty state space;
                \item $\Sigma$ is an alphabet;
                \item $\delta:Q\times\Sigma^h\to Q\times\{0,1\}^{h-1}$ is the \emph{transition function};
                \item $\beta:Q\to\Delta_\Q(\Sigma)$, where $\Delta_\Q(\Sigma)$ denotes the set of all rational-valued discrete probability distributions over $\Sigma$, is the \emph{betting function};
                \item $q_0\in Q$ is the \emph{initial state}; and
                \item $c_0\in[0,\infty)$ is the \emph{initial capital}.
            \end{itemize}
        \end{definition}
        There are $h$ \emph{heads}: $h-1$ \emph{trailing heads}, numbered 1 through $h-1$, and a \emph{leading head}, numbered $h$. All heads start at position $0$, and the gambler proceeds through an input sequence $S\in\Sigma^\omega$ in discrete time steps. In each step, when the leading head is at some position $S[n-1]$ and the gambler is in state $q$, the following events happen simultaneously.
        \begin{itemize}
            \item The gambler places bets on each symbol in $\Sigma$ according to its current state $q$; allocating to each symbol $b\in\Sigma$ the fraction $\beta(q)(b)$ of its current capital.
            \item The current state is updated to the first component of $\delta(q,(b_1,\ldots,b_h))$, where each $b_i$ is the symbol in $S$ at the position of head $i$.
            \item The leading head moves to the right, and some subset of the trailing heads also move. Specifically, head $i$ will move right if and only if the $i$\textsuperscript{th} bit of the second component of $\delta(q,(b_1,\ldots,b_h))$ is 1.
        \end{itemize}
        When the sequence $S$ is fixed, for each $n$ we write $q_n$ for the state of the gambler and $\pi_i(n)$ for the position of head $i$ after $n$ of these time steps.

    \subsection{Oblivious Multi-Head Finite-State Gamblers}

        Here we characterize oblivious multi-head finite-state gamblers, as defined in~\cite{mhfsd}, formally as a special case of adaptive finite-state gamblers. In particular, the parts of the gambler that govern its movements need to be decoupled from the parts that read symbols in the input sequence.
        \begin{definition}
            Let $h$ be a positive integer. An adaptive $h$-head finite-state gambler
            \[G=(Q,\Sigma,\delta,\beta,q_0,c_0)\]
            is an \emph{oblivious $h$-head finite-state gambler} if there are sets $P$ and $T$ and functions $\delta_P:P\times \Sigma^h\to P$, $\delta_T:T\to T$, and $\mu:T\to\{0,1\}^{h-1}$ such that $Q=P\times T$ and, for all $(p,t)\in Q$ and $(b_1,\ldots,b_h)\in\Sigma^h$,
            \[\delta((p,t),(b_1,\ldots,b_h))=\big((\delta_P(p,(b_1,\ldots,b_h)),\delta_T(t)),\mu(t)\big).\]
        \end{definition}
        Thus, the component of the state from $T$ governs the head movements, while only the component from $P$ depends on the symbols that are read, and the head movements are therefore independent of the input sequence.

        As observed in~\cite{mhfsd}, the obliviousness condition implies that each trailing head moves at what is essentially a constant speed, in the following sense.
        \begin{observation}[\cite{mhfsd}]\label{obs:speed}
            In an oblivious $h$-FSG, for each trailing head $i$, there is a speed parameter $\eta_i$ such that, for all $n$, $\eta_i n-|T|\leq \pi_i(n)\leq\eta_i n+|T|$. Since heads with speed parameters 0 or 1 can be simulated by the leading head using finitely many states, one can assume without loss of generality that $\eta_i\in(0,1)$.
        \end{observation}

\section{Algorithmic Information Preliminaries}\label{sec:prelim}
    Fix an alphabet $\Sigma$. We are primarily concerned with how predictable sequences are to multi-head finite-state gamblers, and we will quantify this predictability using $s$-gales. Predictability is closely tied to compressibility, so we will also reason about the compressibility of strings in these sequences, formalized using Kolmogorov complexity.

    \subsection{$s$-Gales and Martingales}
        \begin{definition}
            For $s\in[0,\infty)$, an \emph{$s$-gale} is a function $d:\Sigma^*\to[0,\infty)$ such that, for all $w\in\Sigma^*$,
            \[d(w) = |\Sigma|^{-s}\sum_{b\in \Sigma} d(wb).\]
            A \emph{martingale} is a 1-gale.
        \end{definition}
        Intuitively, we regard an $s$-gale as representing the capital of a gambler who places successive bets on symbols in a sequence. After gambling on some string $w$, the gambler has capital $d(w)$, and it now places another bet, allocates portions of its capital according to its beliefs about the next symbol. If the gambler allocates some portion $\alpha d(w)$ to the symbol $b$ and $b$ is the actual next symbol, the gambler's new capital is $d(wb)=|\Sigma|^s\alpha d(w)$. A martingale is considered the ``fair'' version of this setup; in a martingale, if the gambler makes a uniform bet, allocating $1/|\Sigma|d(w)$ to each symbol $b$, then it is guaranteed to maintain constant capital. The lower $s$ is, the less favorable the environment is for the gambler. The gambler's goal is to have unbounded capital, either infinitely often or cofinitely often; these outcomes characterize \emph{success} and \emph{strong success}, respectively.
    
        \begin{definition}
            An $s$-gale $d$ \emph{succeeds} on a sequence $S\in\Sigma^{\omega}$, and we write $S\in S^\infty(d)$, if
            \[\limsup_{n \to \infty} d(S[0..n-1]) = \infty.\]
            We say $d$ \emph{succeeds strongly} on $S$, and write $S\in S^\infty_{\strong}(d)$, if
            \[\liminf_{n \to \infty} d(S[0..n-1]) = \infty.\]
        \end{definition}
    
        As in~\cite{mhfsd}, given a positive integer $h$ and an oblivious or adaptive $h$-head finite-state gambler $G$, the \emph{martingale of $G$} is given by $d_G(\lambda)=c_0$ and, for each string $w\in\Sigma^*$ and symbol $b\in\Sigma$,
        \[d_G(wb)=|\Sigma|d_G(w)\beta(q_{|w|})(b),\]
        where $q_{|w|}$ is the gambler's state after reading $w$. Given any non-negative real $s$, the \emph{$s$-gale of $G$} is defined, for each string $w\in\Sigma^*$, by
        \[d_G^{(s)}(w)=|\Sigma|^{(s-1)|w|}d_G(w).\]
        Note that the gambler does not have access to its martingale (or other $s$-gale) value, which might involve an unbounded amount of information.

    \subsection{Kolmogorov Complexity and Martin-L\"of Randomness}
        Fix a universal prefix Turing machine $\mathcal{U}$. For all strings $\sigma,\tau\in\Sigma^*$, the \emph{conditional (prefix) Kolmogorov complexity} of $\sigma$ given $\tau$ is
        \[K(\sigma\mid\tau)=\min\{|\pi|\mid \pi\in\{0,1\}^*\text{ and }\mathcal{U}(\pi,\tau)=\sigma\},\]
        and the \emph{(prefix) Kolmogorov complexity} of $\sigma$ is $K(\sigma)=K(\sigma\mid\lambda)$.
        
        We will use only a few basic properties of Kolmogorov complexity; more details and background can be found in~\cite{LiVit19}. First, the complexity of a string is never much more than its length multiplied by the log of the alphabet size: for all $x\in\Sigma^*$, $K(x)\leq |x|\log|\Sigma|+O(\log|x|)$. Second, we have \emph{symmetry of information}: for all $x,y\in\Sigma^*$,
        \[K(x,y)-O(1)\leq K(x\mid y)+K(y)\leq K(x,y)+O(\log|y|).\]
        Third, computable functions cannot significantly increase complexity: if $f:\Sigma^*\to\Sigma^*$ is a computable function, then for all $\sigma,\tau\in\Sigma^*$, applying $f$ cannot significantly add to the information content of $x$, in the sense that $K(f(\sigma))\leq K(\sigma)+O(1)$, and by similar reasoning,
        \[K(\tau\mid f(\sigma))\geq K(\tau\mid \sigma)-O(1).\]
    
        A sequence $S\in\Sigma^\omega$ is \emph{Martin-L\"of random}~\cite{Mart66} if $K(S[0..n-1])\geq n\log|\Sigma|-O(1)$ holds for all $n$. Such a sequence is essentially \emph{incompressible}, and it is well-known that almost all sequences have this property~\cite{DowHir10}.

    \subsection{Using Kolmogorov Complexity to Bound Martingales}

    Intuitively, if a string is difficult to compress, then it is also difficult to predict. Furthermore, if it is difficult to compress \emph{even given the information accessible to the trailing heads}, then it is also difficult to predict with access to that information. This idea is formalized in the following lemma, which will be our main tool for proving upper bounds on the martingales of multi-head finite-state gamblers.
    \begin{lemma}\label{lem:dtok}
        Let $\Sigma$ be an alphabet, $h\geq 1$, $G$ be any adaptive or oblivious $h$-head finite-state gambler over $\Sigma$, $\alpha,\delta\in(0,1)\cap\Q$,  $S\in\Sigma^\omega$, $m,n\in\N$, and let $U$ be any set that includes all positions strictly less than $m$ of the trailing heads while the leading head reads $S[m..n]$, i.e.,
        \[U\supseteq[0,m-1]\cap\bigcup_{i=1}^{h-1}[\pi_i(m),\pi_i(n)].\]
        If $n-m$ is sufficiently large and
        $K(S[[m,n]]\mid S[U])\ge (1-\alpha+\delta)(n-m)\log|\Sigma|$,
        then, letting $x=S[m..n]$ and $y=S[0..m-1]$,
        \begin{equation}\label{eq:nowin}
            \max_{w\sqsubseteq x}\frac{d_G(yw)}{d_G(y)}\leq |\Sigma|^{\alpha\cdot(n-m)}.
        \end{equation}
    \end{lemma}
    This is a direct generalization of Lemma 5.7 in~\cite{mhfsd}, which was stated only for oblivious gamblers and for specific choices of $\Sigma$, $S$, $m$, $U$, $\alpha$, and $\delta$. We closely follow that proof, which uses a standard simulation technique to relate incompressibility and unpredictability; the more general form stated here requires no major changes to the original argument.

    \begin{proof}[Proof of Lemma~\ref{lem:dtok}]
     The function $g:\Sigma^{n-m+1}\to\Q$ defined by
    \[g(z)=\max_{w\sqsubseteq z}\frac{d_G(yw)}{d_G(y)},\]
    is computable given $u=S[U]$ and the state $q$ of $G$ after the leading head has read $y$; as $u$ contains all information about $y$ that the trailing heads could possibly access while the leading head moves through each suffix $w$ of length at most $n-m+1$, it is possible to simulate $G$ to obtain the state sequence and betting ratios needed to compute $g(z)$. (Unlike in~\cite{mhfsd}, $u$ does not necessarily contain \emph{all} information the trailing heads could access. In the adaptive setting, a trailing head might stay very close to the leading head, so a trailing head might reach position $m$ before the leading head reaches position $n$. Thus, $u$ only gives us the trailing heads' information about $y$. This change in the definition of $u$ is harmless, though, as information about $S[m..n]$ is not needed to simulate $G$ on $yw$.) In particular, if $q_{m+i}$ is the gambler's state after the leading head has read $S[0..m+i-1]$, then
    \[g(z)=\max_{w\sqsubseteq z}\prod_{i=0}^{|w|-1}|\Sigma|\beta(q_{m+i})(w[i]).\]

    Let 
    \[Z=\left\{z\in\Sigma^{n-m+1}\;\middle|\;g(z)>|\Sigma|^{\alpha\cdot(n-m)}\right\}.\]
    Since $g$ is computable given $u$ and $q$, there is some Turing machine $M$ that, given input
    \[(j,q,u)\in\N\times Q\times\Sigma^{n+1},\]
    prints $0^mz$, where $z$ is the lexicographically $j$\textsuperscript{th} string in $Z$ (and 0 is the lexicographically first symbol in $\Sigma$).
    
    Before proceeding, we bound $|Z|$. For each $\ell\in\N$, the martingale property gives us
    \[\sum_{w\in\Sigma^\ell}\frac{d_G(yw)}{d_G(y)}=|\Sigma|^\ell,\]
    so
    \[\left|\left\{w\in\Sigma^\ell\;\middle|\;\frac{d_G(yw)}{d_G(y)}>|\Sigma|^{\alpha\cdot(n-m)}\right\}\right|<|\Sigma|^{\ell-\alpha\cdot(n-m)},\]
    and for each string $w$ in this set, there are $|\Sigma|^{n-m-|w|}$ extensions of $w$ in $Z$. Therefore,
    \begin{align*}
        |Z|&\leq\sum_{\ell=0}^{n-m+1}|\Sigma|^{\ell-\alpha\cdot(n-m)}|\Sigma|^{n-m-\ell}\\
        &<(n-m+2)|\Sigma|^{(1-\alpha)(n-m)}.
    \end{align*}
    
    Now assume~\eqref{eq:nowin} does not hold. Then there is some $j$ such that $M(j,q,u)$ prints $0^mx=0^mS[m..n]=S[[m,n]]$. Therefore, by our construction of $M$,
    \begin{align*}
        K(S[[m,n]]\mid u)&\leq K(j,q)+O(1)\\
        &\leq K(j)+O(1)\\
        &\leq \log(|Z|)+O(\log n)\\
        &\leq (1-\alpha)(n-m)\log|\Sigma|+O(\log n)\\
        &< (1-\alpha+\delta)(n-m)\log|\Sigma|,
    \end{align*}
    when $n-m$ is sufficiently large. Thus, we have shown the contrapositive of the implication in the theorem statement.
\end{proof}

\section{Adaptive Multi-Head Finite-State Dimensions}\label{sec:dims}
    To define adaptive multi-head finite predimensions and dimensions, we follow the standard template for defining effective dimension notions with $s$-gales, just as Huang, Li, Lutz, and Lutz~\cite{mhfsd} did in defining oblivious multi-head finite predimensions and dimensions. We deviate from the terminology and notation of~\cite{mhfsd} by including \emph{oblivious} explicitly and prepending an \emph{o} to the dimension notations defined in~\cite{mhfsd} to emphasize that these are the oblivious varieties.

    \begin{definition}[\cite{mhfsd}]
        Let $S$ be a sequence. For each positive integer $h$, the \emph{oblivious $h$-head finite-state predimension} of $S$ is
        \[\odimFS{(h)}(S)=\inf\left\{s\in[0,\infty)\;\middle|\; \exists\text{ oblivious $h$-FSG $G$ with }\left(S\in S^\infty\big(d_G^{(s)}\big)\right)\right\},\]
        and the \emph{oblivious $h$-head finite-state strong predimension} of $S$ is
        \[\oDimFS{(h)}(S)=\inf\left\{s\in[0,\infty)\;\middle|\; \exists\text{ oblivious $h$-FSG $G$ with }\left(S\in S_{\strong}^\infty\big(d_G^{(s)}\big)\right)\right\}.\]
        The \emph{oblivious multi-head finite-state dimension} and \emph{oblivious multi-head finite-state strong dimension} of $S$ are
        \[\odimFS{\mh}(S)=\inf_{h\geq 1}\odimFS{(h)}(S)\text{ and }\oDimFS{\mh}(S)=\inf_{h\geq 1}\oDimFS{(h)}(S).\]
    \end{definition}
    The adaptive versions $\adimFS{(h)}(S)$, $\aDimFS{(h)}(S)$, $\adimFS{\mh}(S)$, and $\aDimFS{\mh}(S)$ are defined identically, except with adaptive $h$-FSGs replacing the oblivious $h$-FSGs. Since oblivious $h$-FSGs are a special case of adaptive $h$-FSGs, it is immediate that the adaptive versions of these quantities are bounded above by the oblivious versions.
    
\section{Defining our Family of Sequences}\label{sec:seq}

    In this section we define, for each integer $h\geq 2$, sequences with patterns that can be exploited more thoroughly by adaptive $h$-FSGs than by oblivious $h$-FSGs; we will prove this separation in Sections~\ref{sec:ub} and~\ref{sec:lb}.
    
    Fix any constant $L\geq 1$. For each integer $h\geq 2$, we now define the function
    \[\Phi_h:(\{0,1\}^L)^\omega\to(\{0,1\}^L\cup\{\$\})^\omega.\]
    Here $\$$ is a special symbol outside the original alphabet.
    
    Given any sequence $S\in(\{0,1\}^L)^\omega$, we construct the sequence $\Phi_h(S)$ by modifying $S$ in the following way for each $k\in\N$. Let $s_k=(h+1)^{2k}h$ and $t_k=(h+1)^{2k+1}$, and let $\oplus$ denote the bitwise parity (XOR) operation. First, set $\Phi_h(S)[s_k]=\Phi_h(S)[t_k]=\$$. Second, for each integer $n\in(s_k,t_k)$ such that $n \bmod{(h+1)}=0$, set
    \begin{equation}\label{eq:seqdef}
        \Phi_h(S)[n]=\left(\bigoplus_{i=1}^{h-2} \Phi_h(S)\left[\frac{i}{h+1}n\right]\right)\oplus \Phi_h(S)\left[\frac{h-(k\bmod 2)}{h+1}n\right].
    \end{equation}
    For all other indices $n$, set $\Phi_h(S)[n]=S[n]$.

    The sequence $\Phi_h(S)$ thus oscillates between three modes, with the symbol $\$$ indicating a switch between modes. From each $t_k$ to $s_{k+1}$, the sequence is in one mode: it is identical to $S$, with no additional structure. From each $s_k$ to $t_k$, a $\frac{1}{h+1}$ fraction of the symbols are replaced by the bitwise parity of $h-1$ symbols from earlier in the sequence. This comprises two modes: the relative positions of the first $h-2$ of these earlier symbols are always the same, but the relative position of the rightmost of the symbols differs according to whether $k$ is even or odd. Intuitively, the fact that this last piece of pertinent information moves between two different relative positions is precisely what makes this sequence more suitable for an adaptive gambler than an oblivious one.

    \begin{figure}
        \centering
        \resizebox{\linewidth}{!}{%
        \begin{tikzpicture}[>=Latex,every node/.style={draw=gold,minimum size=8pt,circle,inner sep=0pt,font=\tiny}]
            \node (s0) at (1.585,0) {\$};
            \node (t0) at (2,0) {\$};
            \node (s1) at (5.585,0) {\$};
            \node (t1) at (6,0) {\$};
            \coordinate (s11) at (3.585,0) {};
            \coordinate (t11) at (4,0) {};
            \coordinate (s12) at (4.585,0) {};
            \coordinate (t12) at (5,0) {};
            \node (s2) at (9.585,0) {\$};
            \node (t2) at (10,0) {\$};
            \coordinate (s21) at (7.585,0) {};
            \coordinate (t21) at (8,0) {};
            \coordinate (s23) at (9.17,0) {};
            \coordinate (t23) at (9.585,0) {};
            \node (s3) at (13.585,0) {\$};
            \node (t3) at (14,0) {\$};
            \coordinate (s31) at (11.585,0) {};
            \coordinate (t31) at (12,0) {};
            \coordinate (s32) at (12.585,0) {};
            \coordinate (t32) at (13,0) {};
            \draw[semithick] (0,0) -- (s0);
            \draw[ultra thick,gold] (s0) -- (t0);
            \draw[semithick] (t0) -- (s1);
            \draw[ultra thick,gold] (s1) -- (t1);
            \draw[semithick] (t1) -- (s2);
            \draw[ultra thick,gold] (s2) -- (t2);
            \draw[semithick] (t2) -- (s3);
            \draw[ultra thick,gold] (s3) -- (t3);
            \draw[ultra thick] (s11) -- (t11);
            \draw[ultra thick] (s12) -- (t12);
            \draw[ultra thick] (s21) -- (t21);
            \draw[ultra thick] (s23) -- (s2);
            \draw[ultra thick] (s31) -- (t31);
            \draw[ultra thick] (s32) -- (t32);
            \draw[thick,gold,->] (5.7925,0) to[out=120,in=60,looseness=1] (3.7925,0);
            \draw[thick,gold,->] (5.7925,0) to[out=-120,in=-60,looseness=1.5] (4.7925,0);
            \draw[thick,gold,->] (9.7925,0) to[out=120,in=60,looseness=1] (7.7925,0);
            \draw[thick,gold,->] (9.7925,0) to[out=-90,in=-90,looseness=3] (9.3775,0);
            \draw[thick,gold,->] (13.7925,0) to[out=120,in=60,looseness=1] (11.7925,0);
            \draw[thick,gold,->] (13.7925,0) to[out=-120,in=-60,looseness=1.5] (12.7925,0);
        \end{tikzpicture}}
        \caption{The structure of a sequence $\Phi_3(S)$, shown on a logarithmic scale for indices 0 to $t_3=16384$. In the short intervals between \$ symbols, one out of every four symbols depends on earlier symbols in the sequence. Specifically, if $m\in\N$ and $4m\in (s_k,t_k)$,  $\Phi_3(S)[4m]$ is either $\Phi_3(S)[m]\oplus \Phi_3(S)[2m]$ or $\Phi_3(S)[m]\oplus \Phi_3(S)[3m]$, depending on the parity of $k$. Note that $t_0=s_0+1$, so the first open interval $(s_0,t_0)$ contains no indices.}
\end{figure}

\section{Upper Bound for Adaptive Multi-Head Finite-State Predimensions}\label{sec:ub}
    The self-referential structure imposed by the function $\Phi_h$ defined in Section~\ref{sec:seq} makes every sequence $\Phi_h(S)$ somewhat predictable to an adaptive $h$-head finite-state gambler. In this section, we construct such gamblers to prove the following upper bounds on the adaptive $h$-head finite-state predimensions of such sequences.

    \begin{theorem}\label{thm:aub}
        For every $h\geq 2$ and every sequence $S\in(\{0,1\}^L)^\omega$,
        \[\adimFS{(h)}(\Phi_h(S))\leq  \frac{L}{\log(2^L+1)}\left(1 - \frac{1}{(h+2)h}\right) \]
        and 
        \[\aDimFS{(h)}(\Phi_h(S))\leq  \frac{L}{\log(2^L+1)}\left(1 - \frac{1}{(h+2)(h+1)h^2}\right).\]
    \end{theorem}

    Let $h\geq 2$, $S\in(\{0,1\}^L)^\omega$, $X=\Phi_h(S)$, and $\varepsilon\in(0,1)\cap\Q$. To prove Theorem~\ref{thm:aub}, we construct an adaptive $h$-head finite-state gambler $G$ over the alphabet $\Gamma =\{0,1\}^L \cup \{ \$\}$ for the sequence $X$.
    
    The heads of $G$ will be numbered $1,\ldots,h$, where $h$ is the leading head. Each trailing head $i\leq h-2$ moves at a fixed speed of $\frac{i}{h + 1}$, and the speed of the last trailing head $h-1$ will range between $\frac{h-1}{h + 1}$ and $\frac{h}{h+1}$. The gambler $G$ has five modes whose behaviors can be implemented with finitely many states. These modes will dictate the betting pattern and the speed of head $h-1$. Two modes are for intervals between positions $s_k$ and $t_k$, where a $\frac{1}{h+1}$ fraction of the sequence's symbols are bitwise parities of earlier symbols. In these modes, the trailing heads access the relevant earlier symbols, and $G$ bets accordingly. During two other modes, $G$ is repositioning head $h-1$, either by slowing down or speeding up, to be prepared for the next $s_k$-to-$t_k$ interval. The fifth mode (which we number $-1$ to make the transition rule simpler) is a one-time start mode.
    
    To define the gambler in more detail, we first define probability distributions that $G$ will use for betting. Define $\nu:\Gamma\to[0,1]$ by $\nu(\$)=\varepsilon$ and $\nu(b)=(1-\varepsilon)/2^L$ for all $b\in\{0,1\}^L$. For each symbol $a\in\Gamma$, define the distribution $\chi_a:\Gamma\to[0,1]$ by $\chi_a(b)=1-\varepsilon$ if $b=a$ and $\chi_a(b)=\varepsilon/2^L$ otherwise.

    The gambler always ``hedges'' for the possibility that the symbol $\$$ will appear rather than attempting to predict its appearance. Because the instances of this symbol are exponentially sparse in $X$, this hedge can safely be arbitrarily small.

    \begin{itemize}
        \item In Mode $-1$, head $h-1$ moves right in the first $h-1$ out of every $h+1$ steps, and $G$ bets according to $\nu$.
        \item In Mode 0, head $h-1$ moves right in the first $h-1$ out of every $h+1$ steps. 
        If the current position is not a multiple of $h+1$, then $G$ bets according to $\nu$. If the current position is a multiple of $h+1$, let $b_1,\ldots,b_h$ be the symbols observed by the $h$ heads. If $b_{h-1}=\$$, then $G$ bets according to $\chi_\$$; otherwise, it bets according to $\chi_{\bigoplus_{i=1}^{h-1}b_i}$.
        \item In Mode 1, head $h-1$ \emph{speeds up}, moving right in the first $h^3+h^2-h+1$ out of every $h^3 + 2h^2 -1$ steps, and $G$ bets according to $\nu$.
        \item In Mode 2, head $h-1$ moves right in the first $h-1$ out of every $h+1$ steps, and the betting behavior of $G$ is identical to Mode 0.
        \item In Mode 3, head $h-1$ \emph{slows down}, moving right in the first $h^3-2h$ out of every $h^3 + 2h^2 -1$ steps, and $G$ bets according to $\nu$.
    \end{itemize}
    The gambler $G$ starts in mode $-1$. Whenever the leading head reads a \$ in mode $m$, the gambler moves into mode $(m+1)\bmod 4$; mode $-1$ is never revisited.
        
    \begin{lemma}\label{lem:rho}
        For $n\in\N$, let
        \[\rho(n)=\frac{1}{n}\left|\left\{i< n\;\middle|\;\frac{d_G(X[0..i])}{d_G(X[0..i-1])}=(1-\varepsilon)|\Gamma|\right\}\right|.\]
        Informally, this is the fraction of symbols in $X[0..n-1]$ on which $G$ makes a full winning bet. Let $\rho_1 = \limsup_{n\to\infty}\rho(n)$  and $\rho_2 = \liminf_{n\to\infty}\rho(n)$. Then
        \[\rho_1=\frac{1}{h(h+2)}\text{ and }\rho_2=\frac{1}{h^2(h+1)(h+2)}.\]
    \end{lemma}
    \begin{proof}
        Observe that $\rho(n)$ achieves its local maxima at $t_k$ and local minima at $s_k$, and that the number of correct full bets placed during each $(s_k,t_k)$ interval is $\frac{t_k-s_k}{h+1}\pm O(1)$. Therefore,
        \begin{align*}
            \limsup_{n\to\infty}\rho(n)&=\lim_{k\to\infty}\frac{\sum_{i=0}^k(t_i-s_i)}{(h+1)t_k}\\
            &=\lim_{k\to\infty}\frac{\sum_{i=0}^k(h+1)^{2i}}{(h+1)^{2k+2}}\\
            &=\lim_{k\to\infty}\frac{(h+1)^{2k+2}-1}{((h+1)^2-1)(h+1)^{2k+2}}\\
            &=\frac{1}{(h+2)h},
        \end{align*}
        and
        \begin{align*}
            \liminf_{n\to\infty}\rho(n)&=\lim_{k\to\infty}\frac{\sum_{i=0}^k(t_i-s_i)}{(h+1)s_{k+1}}\\
            &=\lim_{k\to\infty}\frac{\sum_{i=0}^k(h+1)^{2i}}{(h+1)^{2k+3}h}\\
            &=\lim_{k\to\infty}\frac{(h+1)^{2k+2}-1}{((h+1)^2-1)(h+1)^{2k+3}h}\\
            &=\frac{1}{(h+2)(h+1)h^2}.
        \end{align*}
    \end{proof}

    To prove Theorem~\ref{thm:aub}, it suffices to show, for all $\varepsilon>0$, that the $(L/\log(2^L+1)(1-\rho_1)+\varepsilon)$-gale of $G$ succeeds on $X$ and the $(L/\log(2^L+1)(1-\rho_2)+\varepsilon)$-gale of $G$ succeeds strongly on $X$.

    To that end, define
    \[C(n) = \left(\prod_{i=0}^{n-1}\frac{d_G(X[0..i])}{d_G(X[0..i-1])}\right)^{\frac{1}{n}}.\]
    This value is the geometric mean of the rate of growth per step for the martingale of $G$ as it reads $X[0..n-1]$. Let $M_1 = \limsup_{n\to\infty} C(n)$ and $M_2 = \liminf_{n \to \infty} C(n)$. Then
    \[M_1=\left((2^{L}+1)(1 - \varepsilon)\right)^{\rho_1}\cdot\left(\frac{2^L + 1}{2^L}(1 - \varepsilon)\right)^{1-\rho_1}\]
    and
    \[M_2=\left((2^{L}+1)(1 - \varepsilon)\right)^{\rho_2}\cdot\left(\frac{2^L + 1}{2^L}(1 - \varepsilon)\right)^{1-\rho_2},\]
    where $\rho_1$ and $\rho_2$ are as defined in Lemma~\ref{lem:rho}. Note that we can ignore the contribution of indices where the symbol is $\$$, since this symbol appears with asymptotic density 0.

    Now, consider the $s$-gale induced by the betting behavior of $G$. The geometric mean of its rate of growth per step is
    \[C(n) \cdot (2^{L}+1)^{s-1}.\]
    The $s$-gale thus succeeds whenever $M_1 \cdot (2^{L}+1)^{s-1} > 1$ and succeeds strongly whenever $M_2 \cdot (2^{L}+1)^{s-1} > 1$. Letting $L'=\log(2^L+1)$, we can rewrite these conditions as success whenever
    \begin{align*}
        s &>1- \frac{\log(M_1)}{L'}\\
        &=\frac{L}{L'}(1-\rho_1)-\frac{\log(1 - \varepsilon)}{L'}
    \end{align*}
    and strong success whenever
    \begin{align*}
        s & > 1-\frac{\log(M_2)}{L'} \\
        & = \frac{L}{L'}(1-\rho_2)-\frac{\log(1 - \varepsilon)}{L'}.
    \end{align*}
            
    Since these hold for all $\varepsilon > 0$, it follows that
    \[\adimFS{(h)}(\Phi_h(S))\leq  \frac{L}{L'}(1 - \rho_1)\]
    and
    \[\aDimFS{(h)}(\Phi_h(S))\leq  \frac{L}{L'}(1-\rho_2).\]
    Using Lemma~\ref{lem:rho}, we get that
    \[\adimFS{(h)}(\Phi_h(S))\leq  \frac{L}{L'}\left(1 - \frac{1}{h(h+2)}\right)\]
    and
    \[\aDimFS{(h)}(\Phi_h(S))\leq  \frac{L}{L'}\left(1 - \frac{1}{h^2(h+1)(h+2)}\right),\]
    as desired.
    
\section{Lower Bound for Oblivious Multi-Head Finite-State Predimensions}\label{sec:lb}

    In this section we complete our separation between adaptive and oblivious multi-head finite-state gamblers by showing that for all $h\geq 2$, no oblivious $h$-FSG can match the performance of the adaptive $h$-FSG defined in Section~\ref{sec:ub} on the same family of sequences. Formally, we use Kolmogorov complexity arguments to establish the following lower bounds on the oblivious $h$-head finite-state predimension and strong predimension of these sequences.
    \begin{theorem}\label{thm:olb}
        For every $h\geq 2$ and every Martin-L\"of random sequence $R\in(\{0,1\}^L)^\omega$,
        \[\odimFS{(h)}(\Phi_h(R))\geq \frac{L}{\log(2^L+1)}\left(1-\frac{(h+1)^2}{(h+1)^4-1}\right)\]
        and 
        \[\oDimFS{(h)}(\Phi_h(R))\geq \frac{L}{\log(2^L+1)}\left(1-\frac{1}{(h+1)h\cdot ((h+1)^4-1)}\right).\]
    \end{theorem}

    To prove this theorem, let $h\geq 2$, let $R\in(\{0,1\}^L)^\omega$ be Martin-L\"of random, and let $X=\Phi_h(R)$. Let $\Gamma=\{0,1\}^L\cup\{\$\}$, consider any oblivious $h$-FSG $G$ over the alphabet $\Gamma$, and let $H=\{\eta_1,\ldots,\eta_h\}$ be the set of head speeds, where $0<\eta_1<\ldots<\eta_h=1$. For each $i\in\{1,\ldots,h\}$ and $m,n\in\N$, let $a_i(m)=\lfloor\eta_im\rfloor-|T|$ and $b_i(n)=\lceil\eta_in\rceil+|T|$, and define the set
    \[U(m,n)=\bigcup_{i=1}^{h-1}[a_i(m),b_i(n)].\]
    By Observation~\ref{obs:speed}, this $U(m,n)$ is a superset of the positions of the trailing heads while the leading head reads $X[m..n]$, suitable for use in Lemma~\ref{lem:dtok}. In order to apply that lemma, we now establish conditional Kolmogorov complexity bounds.

    \begin{lemma}\label{lem:olb-comp}
        If $j\in\{1,\ldots,h-1\}$ satisfies $\frac{j}{h+1}\not\in H$, then there is some $\gamma\in(0,1)$ satisfying the following property for all sufficiently large $k$: For all integers $n\in (s_k,s_{k+1}]$ and $m\in[\gamma n,n]$, we have $\max U(m,n)<m$ and, if any of $j\leq h-2$, $j=h-(k\bmod 2)$, or $m\geq t_k$ hold, then
        \[K(X[[m,n]]\mid X[U(m,n)])\geq (n-m)L-O(\log n).\]
    \end{lemma}
    
    To see why this lemma is true, let $\eta=\frac{j}{h+1}$ and assume $\eta\not\in H$. Intuitively, during the interval where the gambler bets on symbols in some string $X[m..n]$, it might be helpful to access some of the symbols in $X[\eta m..\eta n]$. This will be the case if the index range $[m,n]$ includes part of the sequence that uses parity calculations involving symbols with indices in the range $[\eta m,\eta n]$---that is, some index range $[s_k,t_k]$ where either $j\leq h-2$ or $j=h-(k\bmod 2)$.
    
    Each trailing head $i$ will only read symbols within the string $X[a_i(m)..b_i(n)]$, though. We will choose $m$ and $n$ such that, for each trailing head $i$, there is no overlap between $[a_i(m),b_i(n)]$ and $[\eta m,\eta n]$ or $[m,n]$, so that head $i$ does not access the pertinent information during this interval. We will further argue that beyond this non-overlap condition, each $X[a_i(m)..b_i(n)]$ also does not tell the gambler anything about $X[\eta m..\eta n]$ in any subtler ways. Informally, this is because the structure of $X$ dictates that $X[\eta m..\eta n]$ consists entirely of random symbols directly from $R$, each of which is copied at most once in the sequence, and all these copies appear in $X[m..n]$, with which $X[a_i(m)..b_i(n)]$ does not overlap.

    \begin{figure}
    \centering
    \resizebox{\linewidth}{!}{%
    \begin{tikzpicture}[>=Latex,every node/.style={draw,minimum size=10pt,circle,inner sep=0pt,font=\footnotesize}]
        \node (s) at (12,0) {\$};
        \node (t) at (14.4,0) {\$};
        \node[draw=gold] (m) at (13,0) {$m$};
        \node[draw=gold] (n) at (14,0) {$n$};
        \coordinate (c) at (13.5,0) {};
        \coordinate (m1) at (2.167,0) {};
        \coordinate (c1) at (2.25,0) {};
        \coordinate (n1) at (2.333,0) {};
        \coordinate (m2) at (4.333,0) {};
        \coordinate (c2) at (4.5,0) {};
        \coordinate (n2) at (4.667,0) {};
        \coordinate (m3) at (6.5,0) {};
        \coordinate (c3) at (6.75,0) {};
        \coordinate (n3) at (7,0) {};
        \coordinate (m4) at (8.667,0) {};
        \coordinate (c4) at (9,0) {};
        \coordinate (n4) at (9.333,0) {};
        \coordinate (m5) at (10.833,0) {};
        \coordinate (c5) at (11.25,0) {};
        \coordinate (n5) at (11.667,0) {};
        \draw[semithick] (0,0) to (m1);
        \draw[ultra thick] (m1) to (n1);
        \draw[semithick] (n1) to (m2);
        \draw[ultra thick] (m2) to (n2);
        \draw[semithick] (n2) to (m3);
        \draw[ultra thick] (m3) to (n3);
        \draw[semithick] (n3) to (m4);
        \draw[ultra thick] (m4) to (n4);
        \draw[semithick] (n4) to (m5);
        \draw[ultra thick] (m5) to (n5);
        \draw[semithick] (n5) to (s);
        \draw[semithick] (s) to (m);
        \draw[ultra thick,gold] (m) to (n);
        \draw[semithick] (n) to (t);
        \draw[thick,gold,->] (c) to[out=135,in=45] (c1);
        \draw[thick,gold,->] (c) to[out=135,in=45] (c2);
        \draw[thick,gold,->] (c) to[out=135,in=45] (c3);
        \draw[thick,gold,->,dashed] (c) to[out=135,in=45] (c4);
        \draw[thick,gold,->,dashed] (c) to[out=135,in=45] (c5);
    \end{tikzpicture}}
    \caption{Consider $\Phi_5(R)$. During an interval $[m,n]\subseteq(s_k,t_k)$, the formula~\eqref{eq:seqdef} applies to one out of every six symbols. The interval $[m,n]$ thereby partially depends on the intervals $\left[\frac{m}{6},\frac{n}{6}\right],\left[\frac{m}{3},\frac{n}{3}\right],\left[\frac{m}{2},\frac{n}{2}\right]$, and either $\left[\frac{2m}{3},\frac{2n}{3}\right]$ or $\left[\frac{5m}{6},\frac{5n}{6}\right]$, depending on the parity of $k$. When $n$ is large and $m$ is close to $n$, the gaps between these five intervals are large, so they cannot all be visited by four trailing heads while the leading head is in $[m,n]$; there must be some unvisited interval $\left[\frac{jm}{6},\frac{jn}{6}\right]$ during this period. If $j=4$, then the gambler might still have all relevant information when $k$ is odd, and similarly for $j=5$ and $k$ even. If $j\leq 3$, though, then the gambler is missing pertinent information about almost all symbols to which~\eqref{eq:seqdef} applies.}
\end{figure}

    \begin{proof}[Proof of Lemma~\ref{lem:olb-comp}]
        More formally, using the convention $\max\emptyset=0$, define $\eta_<=\max(H\cap [0,\eta))$ and $\eta_>=\min(H\cap(\eta,1])$, the head speeds immediately below and above $\eta$, respectively. Let
        \begin{equation*}\label{eq:gamma}
            \gamma\in\left(\sqrt{\max\{\eta_{h-1},\eta_</\eta,\eta/\eta_>\}},1\right),
        \end{equation*}
        noting that this interval is nonempty. Then for all sufficiently large $n$, all $m\in[\gamma n,n]$, and all trailing heads $i$, we have $\eta_i n+|T|< \gamma n< m$ and either $\eta_i n+|T|<\eta m$ or $\eta_i m-|T|>\eta n$. By Observation~\ref{obs:speed}, we also have $a_i(m)\geq\eta_im-|T|$ and $b_i(n)\leq\eta_in+|T|$. Thus,
        \begin{equation}\label{eq:disjoint}
            [m,n]\cap[a_i,b_i]=\emptyset\text{ and }[\eta m,\eta n]\cap [a_i,b_i]=\emptyset
        \end{equation}
        hold for all trailing heads $i$.

        Hence, let $k$ be sufficiently large, let $n\in (s_k,s_{k+1}]$ and $m\in [\gamma n,n]$, and assume at least one of the conditions $j\leq h-2$, $j=h-(k\bmod 2)$, or $m\geq t_k$ holds. Define the sets
        \[V=\left\{\frac{\ell j}{h+1}\;\middle|\; \ell\in(s_k,t_k)\cap[m,n]\text{ and }\ell\bmod (h+1)=0\right\}\]
        and
        \[W=[0,m]\setminus V.\]
        Note that if $[m,n]$ is a subset of $[t_k,s_{k+1}]$---an interval where no parity calculations occur---then $W=[0,m]$.

        Define the strings $u=X[U]$, $u_h=X[[m,n]]$, $w=X[W]$, and $y=X[0..n]$. It follows from~\eqref{eq:disjoint} that $U\subseteq W$, so the mapping $w\mapsto u$ is computable. Therefore,
        \begin{align*}
            K(u_h\mid u)&\geq K(u_h\mid w)-O(1)\\
            &\geq K(u_h,w)-K(w)-O(1)\tag{symmetry of information}.
        \end{align*}
        Furthermore, the mapping $(u_h,w)\mapsto y$ is computable. To see this, observe that for each index $\ell\in W$, we have $X[\ell]=W[\ell]$. For each other index $\ell\in[0..n]$, we have $\ell\in V$. In this case~\eqref{eq:seqdef} gives us
        \begin{align*}
            X[\ell]&=\left(\bigoplus_{\substack{1\leq i\leq h+1\\i\neq j}}X\left[\frac{i}{j}\ell\right]\right)\oplus X\left[\frac{h-1+(k\bmod 2)}{j}\ell\right]\\
            &=\left(\bigoplus_{\substack{1\leq i\leq h\\i\neq j}}W\left[\frac{i}{j}\ell\right]\right) \oplus X\left[\frac{h+1}{j}\ell\right]\oplus W\left[\frac{h-1+(k\bmod 2)}{j}\ell\right],
        \end{align*}
        and since $\frac{h+1}{j}\ell\in[m,n]$, the symbol $X\left[\frac{h+1}{j}\ell\right]$ is given by $u_h$. Therefore,
        \[K(u_h,w)\geq K(y)-O(1).\]
        Putting this together, we have
        \begin{equation}\label{eq:uh_given_u}
            K(u_h\mid u)\geq K(y)-K(w)-O(1).
        \end{equation}

        We now bound $K(y)$. Let
        \[A=[0,n]\cap\bigcup_{i\in\N}\big(\{s_i,t_i\}\cup\{\ell\in(s_i,t_i)\mid \ell\bmod(h+1)=0\}\big);\]        
        this is the set of all indices of symbols in $R[0..n]$ that were overwritten by the function $\Phi_h$ to change $R[0..n]$ to $y$. Then the mapping $(y,R[A])\mapsto R[0..n]$ is computable, so
        \begin{align*}
            K(R[0..n])&\leq K(y,R[A])+O(1)\\
            &\leq K(y)+K(R[A])+O(1),
        \end{align*}
        which we rearrange to
        \begin{equation}\label{eq:yRR}
            K(y)\geq K(R[0..n])-K(R[A])-O(1).
        \end{equation}
        As $R$ is Martin-L\"of random over $\{0,1\}^L$, we have
        \[K(R[0..n])\geq nL-O(1),\]
        and since $R[A]$ is computable from a string of length $|A|$ over the same alphabet,
        \begin{align*}
            K(R[A])&\leq |A|L+O(\log|A|)\\
            &\leq |A|L+O(\log n).
        \end{align*}
        Combining these bounds with~\eqref{eq:yRR} gives us
        \begin{equation}\label{eq:ylb}
            K(y)\geq (n-|A|)L-O(\log n).
        \end{equation}

        Note that $A\cap V=\emptyset$ because whenever $\ell\in(s_k,t_k)$, we have $\frac{\ell j}{h+1}\in (t_{k-1},s_k)$. Since $W=[0,m]\setminus V$ and $W[A]=X[A]$ can be calculated by directly applying~\eqref{eq:seqdef}, the mapping $R[[0,m]\setminus (A\cup V)]\mapsto w$ is computable. Thus,
        \begin{equation}\label{eq:wub}
            K(w)\leq K(R[[0,m]\setminus (A\cup V)])+O(1).
        \end{equation}
        Furthermore, the string $R[[0,m]\setminus (A\cup V)]$ is computable from a string of the non-placeholder symbols present, which is a string of length $|[0,m]\setminus(A\cup V)|$ over the alphabet $\{0,1\}^L$, plus $O(1)$ additional information for $j$, from which the locations of the placeholder bits can be inferred.  Therefore,
        \[K(w)\leq |[0,m]\setminus(A\cup V)|L+O(\log n).\]
        
        We now bound $|[0,m]\setminus(A\cup V)|$. Since $A\cap V=\emptyset$ and $V\subseteq[0,m]$, we ahve
        \begin{align*}
            |[0,m]\setminus(A\cup V)|&=|[0,m]|-|A|+|A\cap(m,n]|-|V|\\
            &\leq m-|A|+\frac{|(s_k,t_k)\cap(m,n]|}{h+1}-\frac{|(s_k,t_k)\cap[m,n]|}{h+1}+O(\log n)\\
            &=m-|A|+O(\log n),
        \end{align*}
        by the definitions of $A$ and $V$. Recalling that $L$ is a constant and combining this with~\eqref{eq:uh_given_u},~\eqref{eq:ylb}, and~\eqref{eq:wub} yields
        \[K(u_h\mid u)\geq (n-m)L-O(\log n).\]
        
    \end{proof}
        
    Even if the trailing heads are optimally positioned for some subsequence, at least an $\frac{h}{h+1}$ fraction of the symbols in the sequence are ``fresh'' random symbols from $R$, which the trailing heads' observations say nothing about. The following lemma formalizes this intuition.
    \begin{lemma}\label{lem:highcompweak}
       Let $\gamma$ be as in Lemma~\ref{lem:olb-comp}. If $k$ is sufficiently large, then for all integers $n\in(s_k,s_{k+1}]$ and $m\in[\gamma n,n]$,
        \[K(X[[m,n]]\mid X[U(m,n)])\geq \frac{h}{h+1}(n-m)L-O(\log n).\]
    \end{lemma}
    \begin{proof}
        Let $k$ be sufficiently large, let $n\in(s_k,s_{k+1}]$ and $m\in[\gamma n,n]$.          Note that by Lemma~\ref{lem:olb-comp}, we have $U(m,n) \cap [m,n] = \emptyset$.

        Define the sets
        \[B=\{\ell \in \{0,\ldots,n\}\; \mid\;\ell\bmod (h+1)=0\}\]
        and $C=\{0,\ldots,n\}\setminus B$, and the strings $u = X[U(m,n)]$, $u_h = X[[m,n]]$, $w=R[[0,m]]$, $x = X[[0,m]]$, $y=R[B\cap[m,n]]$, and $z=R[C\cap[m,n]]$.

        \begin{align*}
            K(u_h \mid u) &\geq K(z\mid u)-O(1)\tag{$u_h \mapsto z$ is computable}\\
            &\geq K(z \mid x)-O(1)\tag{$x\mapsto u$ is computable}\\
            &\geq K(z\mid w,y)-O(1)\tag{$(w,y)\mapsto x$ is computable}\\
            &\geq K(z,w,y)-K(w,y)-O(1)\tag{symmetry of information}\\
            &\geq K(R[0..n])-K(w,y)-O(1)\tag{$(z,w,y)\mapsto R[0..n]$ is computable}\\
            &\geq nL-K(w,y)-O(1)\tag{$R$ is Martin-L\"of random}
        \end{align*}

        Now, $w$ includes $m+1$ symbols from $R$, and $y$ includes $\frac{n-m}{h+1}\pm O(1)$ symbols from $R$; the remaining symbols in both strings are just placeholders. Since the positions of the placeholders follow a simple pattern, this means $(w,y)$ can be computed from a string of the $m+\frac{n-m}{h+1}\pm O(1)$ non-placeholder symbols involved. As these symbols come from the alphabet $\{0,1\}^L$, this implies
        \begin{align*}
            K(w,y)&\leq \left(m+\frac{n-m}{h+1}\right)L+O\left(\log\left(m+1+\frac{n-m}{h+1}\right)\right)\\
            &\leq \left(\frac{n+hm}{h+1}\right)L+O(\log n).
        \end{align*} 
        Combining the above inequalities, we have
        \[K(u_h\mid u)\geq \frac{h}{h+1}(n-m)L-O(\log n).\]   
    \end{proof}

    With these conditional Kolmogorov complexity bounds in hand, we are now ready to complete the proof of Theorem~\ref{thm:olb}. There are two cases to consider.

    \emph{Case 1:} Suppose $\frac{h}{h+1}\not\in H$. Define the function $\beta:\N\to[0,1]$ by
    \[\beta(n)=\frac{1}{n}\sum_{\text{odd }k}\left|[s_k..t_k]\cap [0..n-1]\right|.\]
    Informally, this is the density in $X[0..n-1]$ of intervals on which $G$ \emph{might} gamble more successfully, in case $\frac{h-1}{h+1}\in H$.
    
    Observe that $\beta(n)$ reaches its local minima at positions $s_{2k+1}$ and local maxima at positions $t_{2k+1}$. Thus,
    \begin{align}
        \delta_1&:=\limsup_{n\to\infty}\frac{\beta(n)}{h+1}\notag\\
        &=\lim_{k\to\infty}\frac{\sum_{i=0}^k (t_{2i+1}-s_{2i+1})}{t_{2k+1}(h+1)}\notag\\
        &=\lim_{k\to\infty}\frac{\sum_{i=0}^k (h+1)^{4i+2}}{(h+1)^{4k+4}}\notag\\
        &=\lim_{k\to\infty}\frac{\sum_{i=0}^k (h+1)^{4i}}{(h+1)^{4k+2}}\notag\\
        &=\lim_{k\to\infty}\frac{(h+1)^{4k+4}-1}{(h+1)^{4k+2}((h+1)^4-1)}\notag\\
        &=\lim_{k\to\infty}\frac{(h+1)^2-(h+1)^{-4k-1}}{(h+1)^4-1}\notag\\
        &=\frac{(h+1)^2}{(h+1)^4-1},\label{eq:limsupbeta}
    \end{align}
    and
    \begin{align}
        \delta_2&:=\liminf_{n\to\infty}\frac{\beta(n)}{h+1}\notag\\
        &=\lim_{k\to\infty}\frac{\sum_{i=0}^k (t_{2i+1}-s_{2i+1})}{s_{2k+3}(h+1)}\notag\\
        &=\lim_{k\to\infty}\frac{\sum_{i=0}^k (h+1)^{4i+2}}{(h+1)^{4k+7}h}\notag\\
        &=\lim_{k\to\infty}\frac{\sum_{i=0}^k (h+1)^{4i}}{(h+1)^{4k+5}h}\notag\\
        &=\lim_{k\to\infty}\frac{(h+1)^{4k+4}-1}{(h+1)^{4k+5}h\cdot ((h+1)^4-1)}\notag\\
        &=\lim_{k\to\infty}\frac{1-(h+1)^{-4k-4}}{h\cdot (h+1) ((h+1)^4-1)}\notag\\
        &=\frac{1}{h\cdot (h+1)((h+1)^4-1)}.\label{eq:liminfbeta}
    \end{align}
    
    Let $\varepsilon<1-\max\{L/\log|\Gamma|,\gamma\}$ be rational. We now use Lemmas~\ref{lem:dtok},~\ref{lem:olb-comp}, and~\ref{lem:highcompweak} to show that, for all sufficiently large $n$,
    \begin{equation}\label{eq:dGbd}
        \log d_G(X[0..n-1])\leq \left(\log|\Gamma|-\left(1-\frac{\beta(n)}{h+1}\right)L+(L+4)\varepsilon\right)n.
    \end{equation}
    
    For this, let $L'=\log|\Gamma|$. Define a sequence $(n_i)_{i\in\N}$ of positive integers such that $n_0$ is sufficiently large, and, for all $i\in\N$,
    \begin{enumerate}
        \item $n_i/n_{i+1}\in\left(\gamma,\frac{1}{1+\varepsilon}\right)$, and
        \item for all $k\in\N$, either $[n_i,n_{i+1}-1]\subseteq [s_k,t_k]$ or those two intervals are disjoint.
    \end{enumerate}

    For each $i\in\N$, let $z_i=X[U(n_i,n_{i+1}-1)]$. By Lemma~\ref{lem:olb-comp}, for all $i$ such that $[n_i,n_{i+1}-1]\cap [s_k,t_k]=\emptyset$ for all odd $k$, we have
    \begin{align}
        K(X[[n_i,n_{i+1}-1]]\mid z_i)
        &\geq (n_{i+1}-n_i)\left(L-\frac{\varepsilon}{2}\right)\notag\\
        &\geq\left(\frac{L}{L'}-\frac{\varepsilon}{2}\right)(n_{i+1}-n_i)\log|\Gamma|.\label{eq:olb1}
    \end{align}
    Similarly, by Lemma~\ref{lem:highcompweak}, for all $i$ we have
    \begin{equation}\label{eq:olb2}
        K(X[[n_i,n_{i+1}-1]]\mid z_i)\geq \left(\frac{hL}{(h+1)L'}-\frac{\varepsilon}{2}\right)(n_{i+1}-n_i)L'.
    \end{equation}
    These last two inequalities are each in a suitable form to apply Lemma~\ref{lem:dtok}, with parameters $\delta_1=\delta_2=\varepsilon/2$,
    \[\alpha_1=1-\frac{L}{L'}+\varepsilon,\]
    and
    \[\alpha_2=1-\frac{hL}{(h+1)L'}+\varepsilon,\]
    respectively, for the two applications. Note that our choice of $\varepsilon$ guarantees $\alpha_1,\alpha_2\in(0,1)\cap\Q$, so Lemma~\ref{lem:dtok} applies. Informally, we will use~\eqref{eq:olb1} to get the stronger bound $\alpha_1$ on the gambler's rate of success on most parts of the sequence, and we will use~\eqref{eq:olb2} to get the weaker bound $\alpha_2$ on the gambler's rate of success on the parts of the sequence where its trailing head placement might give it more of an advantage. Since the latter case occurs a $\beta(n)$ fraction of the time, the overall bound will depend on $\alpha_1$, $\alpha_2$, and $\beta(n)$.

    Fix $\ell\in\N$ and $n\in(n_\ell,n_{\ell+1}]$. Intuitively, it would be convenient to have $n=n_i$ for some $i$, or at least so close to some $n_i$ that the difference is negligible. To this end, if $n_\ell/n<\frac{1}{1+\varepsilon}$, then redefine $n_{\ell+1}=n$ and $\ell=\ell+1$; this maintains the two properties we wanted in our sequence of $n_i$, at least for $i<\ell$. In either case, we now have $n_\ell/n\in\left[\frac{1}{1+\varepsilon},1\right]$, so
    \begin{equation}\label{eq:nl}
        d_G(X[0..n-1])\leq d_G(X[0..n_\ell-1])2^{\varepsilon n L'}.
    \end{equation}
    
    We now turn to bounding $d_G(X[0..n_\ell-1])$. Let
    \begin{align*}
        I&=\{i\in\{0,\ldots,\ell-1\}:[n_i,n_{i+1}-1]\cap[s_k,t_k]=\emptyset\text{ for all odd }k\},\\
        J&=\{i\in\{0,\ldots,\ell-1\}:[n_i,n_{i+1}-1]\subseteq[s_k,t_k]\text{ for some odd }k\}.
    \end{align*}
    By our choices of $n_i$, these two sets form a partition of $\{0,\ldots,\ell-1\}$. 
    For each $i\in\N$, let $x_i=X[n_i..n_{i+1}-1]$ and $y_i=X[0..n_i-1]$. Then
    \begin{align*}
        d_G(X[0..n_\ell-1])&= d_G(y_0)\left(\prod_{i=0}^{\ell-1}\frac{d_G(y_ix_i)}{d_G(y_i)}\right)\\
        &=d_G(y_0)\left(\prod_{i\in I}\frac{d_G(y_ix_i)}{d_G(y_i)}\right)\left(\prod_{i\in J}\frac{d_G(y_ix_i)}{d_G(y_i)}\right),
    \end{align*}
    so
    \[\log d_G(X[0..n_\ell-1])=\sum_{i\in I}\log\frac{d_G(y_ix_i)}{d_G(y_i)}+\sum_{i\in J}\log\frac{d_G(y_ix_i)}{d_G(y_i)}+O(1).\]
    As~\eqref{eq:olb1} applies in $I$ and~\eqref{eq:olb2} applies in $J$, applying Lemma~\ref{lem:dtok} gives us
    \begin{align*}
        \log d_G(X[0..n_\ell-1])&\leq\sum_{i\in I}(n_{i+1}-n_i-1)\alpha_1L'+\sum_{i\in J}(n_{i+1}-n_i-1)\alpha_2L'+O(1)\\
        &\leq \sum_{i\in I}(n_{i+1}-n_i)\alpha_1L'+\sum_{i\in J}(n_{i+1}-n_i)\alpha_2L'+O(\log n).
    \end{align*}
    Since
    \[\frac{1}{n_\ell}\sum_{i\in J}(n_{i+1}-n_i)=\beta(n_\ell)\pm \frac{O(1)}{n_\ell},\]
    we can rewrite this as
    \begin{align*}
        \log d_G(X[0..n_\ell-1])&\leq (1-\beta(n_\ell))n_\ell\alpha_1L'+\beta(n_\ell)n_\ell\alpha_2 L' + O(\log n)\\
        &=\left(\alpha_1 L' + \beta(n_\ell)(\alpha_2-\alpha_1)L'\right)n_\ell+O(\log n)\\
        &=\left(L'-L+\frac{\beta(n_\ell)L}{h+1}+L'\varepsilon\right)n_\ell+O(\log n).
    \end{align*}
    Combining this with~\eqref{eq:nl} and noting that
    \[\frac{\beta(n_\ell)}{\beta(n)}\leq\frac{n}{n_\ell}\leq 1+\varepsilon,\]
    we have
    \begin{align*}
        \log d_G(X[0..n-1])&\leq \left(L'-L+\frac{\beta(n_\ell)L}{h+1}+(L'+1)\varepsilon\right)n+O(\log n)\\
        &\leq \left(L'-L+\frac{L}{h+1}(1+\varepsilon)\beta(n)+(L'+1)\varepsilon\right)n+O(\log n)\\
        &\leq \left(L'-L+\frac{L}{h+1}\beta(n)+(L'+2)\varepsilon\right)n+O(\log n)\\
        &\leq \left(L'-L+\frac{L}{h+1}\beta(n)+(L+3)\varepsilon\right)n+O(\log n)\\
        &\leq \left(\log|\Gamma|-\left(1-\frac{\beta(n)}{h+1}\right)L+(L+4)\varepsilon\right)n,
    \end{align*}
    so~\eqref{eq:dGbd} holds for all sufficiently large $n$.

    Therefore,
    \begin{align*}
        &\limsup_{n\to\infty}d_G^{((1-\delta_1)L/\log|\Gamma|-(L+5)\varepsilon)}(X[0..n-1])\\
        &=\limsup_{n\to\infty}|\Gamma|^{((1-\delta_1)L/\log|\Gamma|-(L+5)\varepsilon-1)n}d_G(X[0..n-1])\\
        &=\limsup_{n\to\infty}2^{(1-\delta_1)Ln}|\Gamma|^{-((L+5)\varepsilon+1)n}d_G(X[0..n-1])\\
        &\leq\limsup_{n\to\infty}2^{(1-\delta_1)Ln}|\Gamma|^{-((L+5)\varepsilon+1)n}|\Gamma|2^{-(1-\beta(n)/(h+1))Ln}2^{(L+4)\varepsilon n}\tag{by~\eqref{eq:dGbd}}\\
        &=\limsup_{n\to\infty} 2^{(1-\delta_1)Ln}|\Gamma|^{-((L+5)\varepsilon+1)n}|\Gamma|2^{-(1-\delta_1)Ln}2^{(L+4)\varepsilon n}\tag{by def. of $\delta_1$}\\
        &<\limsup_{n\to\infty}|\Gamma|^{-\varepsilon n}\tag{$|\Gamma|=2^L+1>2$}\\
        &=0,
    \end{align*}
    and similarly,
    \[\liminf_{n\to\infty}d_G^{((1-\delta_2)L/\log|\Gamma|-(L+5)\varepsilon)}(X[0..n-1])=0.\]
    Thus, the $((1-\delta_1)L/\log|\Gamma|-(L+5)\varepsilon)$-gale of $G$ does not succeed on $X$, and the $((1-\delta_2)L/\log|\Gamma|-(L+5)\varepsilon)$-gale of $G$ does not succeed strongly on $X$.

    \emph{Case 2:} If $\frac{h}{h+1}\in H$, then since $|H|=h$ and $h\in H$, there must be some $j\leq h-1$ such that $\frac{j}{h+1}\not\in H$. The rest of the argument is symmetric to Case 1, and we reach the same conclusion about the $((1-\delta_1)L/\log|\Gamma|-(L+5)\varepsilon)$-gale and $((1-\delta_2)L/\log|\Gamma|-(L+5)\varepsilon)$-gale of $G$.

    Every oblivious $h$-FSG $G$ falls into one of the two cases above. Letting $\varepsilon$ approach 0, this implies
    \begin{align*}
        \odimFS{(h)}(X)\geq \frac{(1-\delta_1)L}{\log|\Gamma|}\text{ and }
        \oDimFS{(h)}(X)\geq \frac{(1-\delta_2)L}{\log|\Gamma|},
    \end{align*}
    and applying~\eqref{eq:limsupbeta} and~\eqref{eq:liminfbeta} completes the proof of Theorem~\ref{thm:olb}.
    
    Combined with Theorem~\ref{thm:aub}, this establishes a strict separation between adaptive and oblivious $h$-head predimensions.

\section{A Hierarchy of Adaptive Predimensions}\label{sec:hier}

    In this section we prove our hierarchy theorem, describing, for each $h\geq 1$, sequences whose adaptive $(h+1)$-head predimensions are strictly less than their adaptive $h$-head predimensions. In fact, even the \emph{oblivious} $(h+1)$-head predimensions of these sequences are strictly less than their adaptive $h$-head predimensions.
    
    \begin{theorem}\label{thm:hierarchy}
        For each $h\geq 1$ there is a sequence $Y\in\{0,1\}^\omega$ such that
        \[\oDimFS{(h+1)}(Y)<\adimFS{(h)}(Y).\]
    \end{theorem}
    Since adaptive gamblers generalize oblivious gamblers, this theorem has the following immediate corollary.
    \begin{corollary}
        For each $h\geq 1$ there is a sequence $Y\in\{0,1\}^\omega$ such that
        \[\adimFS{(h+1)}(Y)<\adimFS{(h)}(Y)\text{ and }\aDimFS{(h+1)}(Y)<\aDimFS{(h)}(Y).\]
    \end{corollary}

    This hierarchy is analogous to the oblivious predimension hierarchy theorem of Huang, Li, Lutz, and Lutz~\cite{mhfsd}. In fact, we prove Theorem~\ref{thm:hierarchy} using the same family of sequences used in~\cite{mhfsd}, which is surprising in light of the separation between adaptive and oblivious predimensions established by Theorems~\ref{thm:aub} and~\ref{thm:olb}. We now describe that family of sequences, which is constructed by starting from Martin-L\"of random sequences and then imposing structures of backward-looking self-reference onto these sequences. Prime numbers play an important role in keeping these structures tree-like, so that paths of self-reference are unique. For each positive integer $k$, let $p_k$ denote the $k$\textsuperscript{th} prime number, and define $\nu_k:\Z\to\N$, where $\nu_k(n)$ is the multiplicity of $p_k$ in $n$. Let $p_0=1$.
    
    \begin{definition}[\cite{mhfsd}]
        For each positive integer $h$, define the function
        $F_{h+1}:\{0,1\}^\omega\to\{0,1\}^\omega$
        by, for each sequence $S\in\{0,1\}^\omega$,
        $F_{h+1}(S)[0]=0$,
        and, for all $q,r\geq 0$ such that $r<p_{h+1}$,
        \[F_{h+1}(S)[qp_{h+1}+r]=\begin{cases}\bigoplus_{k=1}^h F_{h+1}(S)[qp_k]&\text{if }r=0\\S[q\cdot(p_{h+1}-1)+r]&\text{if }r>0,\end{cases}\]
        where $\oplus$ denotes the parity operation.
    \end{definition}

    \begin{figure}
        \centering
        \input{8def3head1}
        \caption{After the first bit $F_3(S)[0]$, four out of every five bits in $F_3(S)$ are drawn directly from $S$. Each remaining bit in $F_3(S)$ is $F_3(S)[5n]$ for some $n\geq 1$ and is calculated as $F_3(S)[2n]\oplus F_3(S)[3n]$.}
    \end{figure}

    \begin{figure}
        \centering
        \begin{tikzpicture}[x=0.75pt,y=0.75pt,yscale=-1,xscale=1]

\draw [line width=1.5]    (90,89.83) -- (482,89.83) ;
\draw  [fill=gold  ,fill opacity=1 ] (86,89.83) .. controls (86,92.05) and (87.79,93.84) .. (90,93.84) .. controls (92.21,93.84) and (94,92.05) .. (94,89.83) .. controls (94,87.61) and (92.21,85.82) .. (90,85.82) .. controls (87.79,85.82) and (86,87.61) .. (86,89.83) -- cycle ;
\draw  [fill=gold  ,fill opacity=1 ] (133.87,89.41) .. controls (133.87,91.63) and (135.66,93.43) .. (137.87,93.43) .. controls (140.09,93.43) and (141.88,91.63) .. (141.88,89.41) .. controls (141.88,87.2) and (140.09,85.4) .. (137.87,85.4) .. controls (135.66,85.4) and (133.87,87.2) .. (133.87,89.41) -- cycle ;
\draw  [fill=gold  ,fill opacity=1 ] (157.87,89.41) .. controls (157.87,91.63) and (159.66,93.43) .. (161.87,93.43) .. controls (164.09,93.43) and (165.88,91.63) .. (165.88,89.41) .. controls (165.88,87.2) and (164.09,85.4) .. (161.87,85.4) .. controls (159.66,85.4) and (157.87,87.2) .. (157.87,89.41) -- cycle ;
\draw  [fill=gold  ,fill opacity=1 ] (193.87,89.41) .. controls (193.87,91.63) and (195.66,93.43) .. (197.87,93.43) .. controls (200.09,93.43) and (201.88,91.63) .. (201.88,89.41) .. controls (201.88,87.2) and (200.09,85.4) .. (197.87,85.4) .. controls (195.66,85.4) and (193.87,87.2) .. (193.87,89.41) -- cycle ;
\draw  [fill=gold  ,fill opacity=1 ] (204.87,89.41) .. controls (204.87,91.63) and (206.66,93.43) .. (208.87,93.43) .. controls (211.09,93.43) and (212.88,91.63) .. (212.88,89.41) .. controls (212.88,87.2) and (211.09,85.4) .. (208.87,85.4) .. controls (206.66,85.4) and (204.87,87.2) .. (204.87,89.41) -- cycle ;
\draw  [fill=gold  ,fill opacity=1 ] (264.87,89.41) .. controls (264.87,91.63) and (266.66,93.43) .. (268.87,93.43) .. controls (271.09,93.43) and (272.88,91.63) .. (272.88,89.41) .. controls (272.88,87.2) and (271.09,85.4) .. (268.87,85.4) .. controls (266.66,85.4) and (264.87,87.2) .. (264.87,89.41) -- cycle ;
\draw  [fill=gold  ,fill opacity=1 ] (384.87,89.41) .. controls (384.87,91.63) and (386.66,93.43) .. (388.87,93.43) .. controls (391.09,93.43) and (392.88,91.63) .. (392.88,89.41) .. controls (392.88,87.2) and (391.09,85.4) .. (388.87,85.4) .. controls (386.66,85.4) and (384.87,87.2) .. (384.87,89.41) -- cycle ;
\draw [line width=1.5]    (90,168.83) -- (482,168.83) ;
\draw  [fill=gold  ,fill opacity=1 ] (86,168.83) .. controls (86,171.05) and (87.79,172.84) .. (90,172.84) .. controls (92.21,172.84) and (94,171.05) .. (94,168.83) .. controls (94,166.61) and (92.21,164.82) .. (90,164.82) .. controls (87.79,164.82) and (86,166.61) .. (86,168.83) -- cycle ;
\draw  [fill=gold  ,fill opacity=1 ] (124.87,168.41) .. controls (124.87,170.63) and (126.66,172.43) .. (128.87,172.43) .. controls (131.09,172.43) and (132.88,170.63) .. (132.88,168.41) .. controls (132.88,166.2) and (131.09,164.4) .. (128.87,164.4) .. controls (126.66,164.4) and (124.87,166.2) .. (124.87,168.41) -- cycle ;
\draw  [fill=gold  ,fill opacity=1 ] (144.87,168.41) .. controls (144.87,170.63) and (146.66,172.43) .. (148.87,172.43) .. controls (151.09,172.43) and (152.88,170.63) .. (152.88,168.41) .. controls (152.88,166.2) and (151.09,164.4) .. (148.87,164.4) .. controls (146.66,164.4) and (144.87,166.2) .. (144.87,168.41) -- cycle ;
\draw  [fill=gold  ,fill opacity=1 ] (172.87,168.41) .. controls (172.87,170.63) and (174.66,172.43) .. (176.87,172.43) .. controls (179.09,172.43) and (180.88,170.63) .. (180.88,168.41) .. controls (180.88,166.2) and (179.09,164.4) .. (176.87,164.4) .. controls (174.66,164.4) and (172.87,166.2) .. (172.87,168.41) -- cycle ;
\draw    (273.65,86.1) .. controls (312.54,58.15) and (350.63,61.45) .. (385.87,87.4) ;
\draw [shift={(271.87,87.4)}, rotate = 323.39] [color={rgb, 255:red, 0; green, 0; blue, 0 }  ][line width=0.75]    (8.74,-2.63) .. controls (5.56,-1.12) and (2.65,-0.24) .. (0,0) .. controls (2.65,0.24) and (5.56,1.12) .. (8.74,2.63)   ;
\draw    (213.68,86.08) .. controls (253.79,57.44) and (332.81,37.76) .. (385.87,87.4) ;
\draw [shift={(211.87,87.4)}, rotate = 323.39] [color={rgb, 255:red, 0; green, 0; blue, 0 }  ][line width=0.75]    (8.74,-2.63) .. controls (5.56,-1.12) and (2.65,-0.24) .. (0,0) .. controls (2.65,0.24) and (5.56,1.12) .. (8.74,2.63)   ;
\draw    (141.04,95.07) .. controls (160.39,121.29) and (194.24,116.02) .. (205.87,92.43) ;
\draw [shift={(139.87,93.43)}, rotate = 55.74] [color={rgb, 255:red, 0; green, 0; blue, 0 }  ][line width=0.75]    (8.74,-2.63) .. controls (5.56,-1.12) and (2.65,-0.24) .. (0,0) .. controls (2.65,0.24) and (5.56,1.12) .. (8.74,2.63)   ;
\draw    (165.46,94.05) .. controls (179.98,108.35) and (193.45,106.01) .. (205.87,92.43) ;
\draw [shift={(163.87,92.43)}, rotate = 46.74] [color={rgb, 255:red, 0; green, 0; blue, 0 }  ][line width=0.75]    (8.74,-2.63) .. controls (5.56,-1.12) and (2.65,-0.24) .. (0,0) .. controls (2.65,0.24) and (5.56,1.12) .. (8.74,2.63)   ;
\draw    (202.27,84.91) .. controls (232.78,53.1) and (258.13,72.71) .. (266.87,86.41) ;
\draw [shift={(200.87,86.4)}, rotate = 312.56] [color={rgb, 255:red, 0; green, 0; blue, 0 }  ][line width=0.75]    (8.74,-2.63) .. controls (5.56,-1.12) and (2.65,-0.24) .. (0,0) .. controls (2.65,0.24) and (5.56,1.12) .. (8.74,2.63)   ;
\draw    (164.19,83.22) .. controls (186.39,48.08) and (244.46,38.48) .. (266.87,86.41) ;
\draw [shift={(162.87,85.4)}, rotate = 299.79] [color={rgb, 255:red, 0; green, 0; blue, 0 }  ][line width=0.75]    (8.74,-2.63) .. controls (5.56,-1.12) and (2.65,-0.24) .. (0,0) .. controls (2.65,0.24) and (5.56,1.12) .. (8.74,2.63)   ;
\draw    (178.39,162.47) -- (196.87,93.43) ;
\draw [shift={(177.87,164.4)}, rotate = 284.99] [color={rgb, 255:red, 0; green, 0; blue, 0 }  ][line width=0.75]    (8.74,-2.63) .. controls (5.56,-1.12) and (2.65,-0.24) .. (0,0) .. controls (2.65,0.24) and (5.56,1.12) .. (8.74,2.63)   ;
\draw    (149.21,162.43) -- (160.87,93.43) ;
\draw [shift={(148.87,164.4)}, rotate = 279.6] [color={rgb, 255:red, 0; green, 0; blue, 0 }  ][line width=0.75]    (8.74,-2.63) .. controls (5.56,-1.12) and (2.65,-0.24) .. (0,0) .. controls (2.65,0.24) and (5.56,1.12) .. (8.74,2.63)   ;
\draw    (129.1,162.41) -- (136.87,93.43) ;
\draw [shift={(128.87,164.4)}, rotate = 276.43] [color={rgb, 255:red, 0; green, 0; blue, 0 }  ][line width=0.75]    (8.74,-2.63) .. controls (5.56,-1.12) and (2.65,-0.24) .. (0,0) .. controls (2.65,0.24) and (5.56,1.12) .. (8.74,2.63)   ;

\draw (39.64,69.8) node [anchor=north west][inner sep=0.75pt]   [align=left] {$\displaystyle F_{3}( S)$};
\draw (66.81,150.46) node [anchor=north west][inner sep=0.75pt]   [align=left] {$\displaystyle S$};
\draw (84.28,94.86) node [anchor=north west][inner sep=0.75pt]   [align=left] {0};
\draw (114.87,94.41) node [anchor=north west][inner sep=0.75pt]   [align=left] {48};
\draw (142.88,70.41) node [anchor=north west][inner sep=0.75pt]   [align=left] {72};
\draw (171.87,70.4) node [anchor=north west][inner sep=0.75pt]   [align=left] {108};
\draw (210.87,94.41) node [anchor=north west][inner sep=0.75pt]   [align=left] {120};
\draw (270.87,94.41) node [anchor=north west][inner sep=0.75pt]   [align=left] {180};
\draw (390.87,94.41) node [anchor=north west][inner sep=0.75pt]   [align=left] {300};
\draw (119.28,173.86) node [anchor=north west][inner sep=0.75pt]   [align=left] {39};
\draw (139.28,173.86) node [anchor=north west][inner sep=0.75pt]   [align=left] {58};
\draw (167.28,173.86) node [anchor=north west][inner sep=0.75pt]   [align=left] {87};
\draw (84,173.83) node [anchor=north west][inner sep=0.75pt]   [align=left] {0};

\end{tikzpicture}
        \caption{When the multiplicity of $p_{h+1}$ in $n$ is greater than 1, cancellation can appear in the dependency structure of $F_{h+1}(S)$. For example, $F_3(S)[300]$ nominally depends on $F_3(S)[72]$ via both $F_3(S)[120]$ and $F_3(S)[180]$, but these dependencies will cancel each other out when the parity operation is applied.}\label{fig:8def3head2}
    \end{figure}
 
    Now let $h$ be any positive integer, $R$ be any Martin-L\"of random binary sequence, and $Y=F_{h+1}(R)$. Informally, every bit $Y[n]$ in the sequence is either the parity of exactly $h$ bits from earlier in the sequence (when $n$ is a multiple of $p_{h+1}$) or a ``fresh'' random bit from $R$ that no effective gambler can hope to predict. In the former case, $Y[n]$ is always \emph{sensitive} to each of those $h$ earlier bits, so the gambler would need information about all $h$ of them to successfully predict $Y[n]$. Intuitively, $h+1$ heads (including $h$ trailing heads) are sufficient to access that information and make confident predictions about these parity bits, but with only $h$ heads (including $h-1$ trailing heads), the gambler is always missing something important.
    
    An upper bound with $h+1$ heads comes directly from Lemma 5.4 of~\cite{mhfsd}, which constructed an oblivious $(h+1)$-FSG to establish that
    $\oDimFS{(h+1)}(Y)\leq 1-1/p_{h+1}$.
    Hence, to prove Theorem~\ref{thm:hierarchy}, it suffices to show
    \begin{equation}\label{eq:fulladim}
        \adimFS{(h)}(Y)= 1.
    \end{equation}
    Huang, Li, Lutz, and Lutz~\cite{mhfsd} proved the oblivious analog of~\eqref{eq:fulladim} by arguing that for every oblivious $h$-head finite-state gambler, there is some fixed $j\in\{1,\ldots,h\}$ and some threshold $n_0$ such that, for nearly all $n>n_0$, the gambler has no information about the $j$\textsuperscript{th} of these $h$ bits, and that without this $j$\textsuperscript{th} piece of pertinent information, it has very little ability to predict bits in the sequence. The recursive structure of the sequence makes the details require care (e.g., the parameterization of ``nearly all''), but~\cite{mhfsd} formalized this idea with a conditional Kolmogorov complexity lower bound on substrings of $Y$, from which they inductively derived $\odimFS{(h)}(Y)= 1$.

    Our proof of~\eqref{eq:fulladim} is based on reversing the order of the quantifiers on $j$ and $n_0$ in the argument of~\cite{mhfsd}. With adaptive head movements, we cannot guarantee that there is a \emph{single} $j$ for which the $j$\textsuperscript{th} pertinent piece of information is not available to the gambler, but intuitively, it is enough to say that there is nearly always \emph{some} $j$ for which the gambler lacks that information.
    
    Dropping the assumption of oblivious head movements slightly complicates the argument of~\cite{mhfsd}, which is cast in terms of the fixed speeds of oblivious heads. Those fixed speeds guarantee that the gap between the trailing heads and the leading head grows linearly, which enables reasoning about linear-length strings that have been visited by the leading head but not yet by any trailing heads. In the adaptive setting, this gap can be sublinear and even sublogarithmic, which makes the directly analogous argument too delicate for the logarithmic error terms incurred in our Kolmogorov complexity calculations. This pushes us to instead reason about strings that may also have been visited by trailing heads. Nevertheless, portions of the argument from~\cite{mhfsd} hold unchanged in this setting.
    
    For all $m,n\in\N$ such that $m\leq n$, define
    \[U(m,n)=[0,m-1]\cap\bigcup_{i=1}^{h-1}[\pi_i(m-1),\pi_i(m-1)+n-m+1],\]
    which includes all possible trailing head positions less than $m$ while the leading head reads $n-m+1$ bits after $Y[0..m-1]$.
    
    Fix a positive integer $d$ to bound the recursion depth we consider. For each $1\leq i\leq h$ and all $m,n\in\N$ such that $m\leq n$, define
    \[V_i(m,n)=\left\{(p_i/p_{h+1})^{\nu_{h+1}(k)}k \;\middle|\; k\in[m..n]\text{ and }\nu_{h+1}(k)\leq d\right\}.\]
    Informally, each index in $V_i(m,n)$ is a leaf in the tree of parity operations that calculate some bit in $Y[m..n]$; specifically, it is the leaf reached by taking the $i$\textsuperscript{th} child at each level. Note that $V_i(m,n)$ and $[m..n]$ will generally not be disjoint, as most indices $\ell$ have $\nu_{h+1}(\ell)=0$. An example is shown in Figure~\ref{fig:8treeargument}.

    \begin{figure}
        \centering
        \tikzset{every picture/.style={line width=0.75pt}} 

\begin{tikzpicture}[x=0.75pt,y=0.75pt,yscale=-1,xscale=1]

\draw [color=gold  ,draw opacity=1 ][fill=gold  ,fill opacity=1 ]   (256.35,49.8) -- (230.01,133.47) ;
\draw [shift={(229.41,135.38)}, rotate = 287.47] [color=gold  ,draw opacity=1 ][line width=0.75]    (10.93,-3.29) .. controls (6.95,-1.4) and (3.31,-0.3) .. (0,0) .. controls (3.31,0.3) and (6.95,1.4) .. (10.93,3.29)   ;
\draw  [fill=black  ,fill opacity=1 ] (594.29,58.46) .. controls (594.29,56.5) and (595.88,54.91) .. (597.84,54.91) .. controls (599.8,54.91) and (601.39,56.5) .. (601.39,58.46) .. controls (601.39,60.42) and (599.8,62.01) .. (597.84,62.01) .. controls (595.88,62.01) and (594.29,60.42) .. (594.29,58.46) -- cycle ;
\draw  [fill=black  ,fill opacity=1 ] (194.27,238.93) .. controls (194.27,236.97) and (195.85,235.38) .. (197.81,235.38) .. controls (199.77,235.38) and (201.36,236.97) .. (201.36,238.93) .. controls (201.36,240.89) and (199.77,242.48) .. (197.81,242.48) .. controls (195.85,242.48) and (194.27,240.89) .. (194.27,238.93) -- cycle ;
\draw    (410.69,58.46) -- (597.84,58.46) ;
\draw  [fill=black  ,fill opacity=1 ] (224.87,138.93) .. controls (224.87,136.97) and (226.45,135.38) .. (228.41,135.38) .. controls (230.37,135.38) and (231.96,136.97) .. (231.96,138.93) .. controls (231.96,140.89) and (230.37,142.48) .. (228.41,142.48) .. controls (226.45,142.48) and (224.87,140.89) .. (224.87,138.93) -- cycle ;
\draw [color=gold  ,draw opacity=1 ][fill=gold  ,fill opacity=1 ]   (226.87,142.93) -- (199.39,233.47) ;
\draw [shift={(198.81,235.38)}, rotate = 286.88] [color=gold  ,draw opacity=1 ][line width=0.75]    (10.93,-3.29) .. controls (6.95,-1.4) and (3.31,-0.3) .. (0,0) .. controls (3.31,0.3) and (6.95,1.4) .. (10.93,3.29)   ;
\draw  [fill=black  ,fill opacity=1 ] (549.06,58.46) .. controls (549.06,56.5) and (550.64,54.91) .. (552.6,54.91) .. controls (554.56,54.91) and (556.15,56.5) .. (556.15,58.46) .. controls (556.15,60.42) and (554.56,62.01) .. (552.6,62.01) .. controls (550.64,62.01) and (549.06,60.42) .. (549.06,58.46) -- cycle ;
\draw  [fill=black  ,fill opacity=1 ] (566.91,58.46) .. controls (566.91,56.5) and (568.5,54.91) .. (570.46,54.91) .. controls (572.41,54.91) and (574,56.5) .. (574,58.46) .. controls (574,60.42) and (572.41,62.01) .. (570.46,62.01) .. controls (568.5,62.01) and (566.91,60.42) .. (566.91,58.46) -- cycle ;
\draw  [fill=black  ,fill opacity=1 ] (408.03,58.46) .. controls (408.03,56.5) and (409.62,54.91) .. (411.58,54.91) .. controls (413.54,54.91) and (415.13,56.5) .. (415.13,58.46) .. controls (415.13,60.42) and (413.54,62.01) .. (411.58,62.01) .. controls (409.62,62.01) and (408.03,60.42) .. (408.03,58.46) -- cycle ;
\draw  [fill=black  ,fill opacity=1 ] (325.27,233.93) .. controls (325.27,231.97) and (326.85,230.38) .. (328.81,230.38) .. controls (330.77,230.38) and (332.36,231.97) .. (332.36,233.93) .. controls (332.36,235.89) and (330.77,237.48) .. (328.81,237.48) .. controls (326.85,237.48) and (325.27,235.89) .. (325.27,233.93) -- cycle ;
\draw [color=gold  ,draw opacity=1 ][fill=gold  ,fill opacity=1 ]   (256.35,49.8) -- (291.22,132.16) ;
\draw [shift={(292,134)}, rotate = 247.05] [color=gold  ,draw opacity=1 ][line width=0.75]    (10.93,-3.29) .. controls (6.95,-1.4) and (3.31,-0.3) .. (0,0) .. controls (3.31,0.3) and (6.95,1.4) .. (10.93,3.29)   ;
\draw  [fill=black  ,fill opacity=1 ] (253.35,50.8) .. controls (253.35,48.84) and (254.94,47.25) .. (256.9,47.25) .. controls (258.86,47.25) and (260.45,48.84) .. (260.45,50.8) .. controls (260.45,52.76) and (258.86,54.34) .. (256.9,54.34) .. controls (254.94,54.34) and (253.35,52.76) .. (253.35,50.8) -- cycle ;
\draw  [fill=black  ,fill opacity=1 ] (290.45,138.55) .. controls (290.45,136.59) and (292.04,135) .. (294,135) .. controls (295.96,135) and (297.55,136.59) .. (297.55,138.55) .. controls (297.55,140.51) and (295.96,142.1) .. (294,142.1) .. controls (292.04,142.1) and (290.45,140.51) .. (290.45,138.55) -- cycle ;
\draw [color=gold  ,draw opacity=1 ][fill=gold  ,fill opacity=1 ]   (296,142.1) -- (327.31,227.12) ;
\draw [shift={(328,229)}, rotate = 249.79] [color=gold  ,draw opacity=1 ][line width=0.75]    (10.93,-3.29) .. controls (6.95,-1.4) and (3.31,-0.3) .. (0,0) .. controls (3.31,0.3) and (6.95,1.4) .. (10.93,3.29)   ;

\draw (410.64,131.8) node [anchor=north west][inner sep=0.75pt]   [align=left] {$\displaystyle p_{1} =2,p_{2} =3,p_{3} =5$};
\draw (240.8,23.83) node [anchor=north west][inner sep=0.75pt]   [align=left] {300};
\draw (185.96,128.25) node [anchor=north west][inner sep=0.75pt]   [align=left] { $\displaystyle 120$};
\draw (160.14,229.04) node [anchor=north west][inner sep=0.75pt]   [align=left] { $\displaystyle 48$};
\draw (151,259) node [anchor=north west][inner sep=0.75pt]   [align=left] {$\displaystyle 48\in V_{1}( m,n)$};
\draw (180.78,81.66) node [anchor=north west][inner sep=0.75pt]   [align=left] {{\small  $\displaystyle \times ( 2/5)$}};
\draw (151.07,181.01) node [anchor=north west][inner sep=0.75pt]   [align=left] {{\small  $\displaystyle \times ( 2/5)$}};
\draw (534.72,60.54) node [anchor=north west][inner sep=0.75pt]   [align=left] { $\displaystyle m$};
\draw (583.18,59.76) node [anchor=north west][inner sep=0.75pt]   [align=left] { $\displaystyle n$};
\draw (394.64,59.76) node [anchor=north west][inner sep=0.75pt]   [align=left] { $\displaystyle 0$};
\draw (556.8,37.83) node [anchor=north west][inner sep=0.75pt]   [align=left] {300};
\draw (411.64,158.8) node [anchor=north west][inner sep=0.75pt]   [align=left] {$\displaystyle p_{1} /p_{3} =2/5$};
\draw (412,184.8) node [anchor=north west][inner sep=0.75pt]   [align=left] {$\displaystyle p_{2} /p_{3} =3/5$};
\draw (278.78,80.66) node [anchor=north west][inner sep=0.75pt]   [align=left] {{\small  $\displaystyle \times ( 3/5)$}};
\draw (296.96,128.25) node [anchor=north west][inner sep=0.75pt]   [align=left] { $\displaystyle 180$};
\draw (315.78,180.66) node [anchor=north west][inner sep=0.75pt]   [align=left] {{\small  $\displaystyle \times ( 3/5)$}};
\draw (330.96,224.25) node [anchor=north west][inner sep=0.75pt]   [align=left] { $\displaystyle 108$};
\draw (300,259) node [anchor=north west][inner sep=0.75pt]   [align=left] {$\displaystyle 108\in V_{2}( m,n)$};

\end{tikzpicture}
        \caption{
            Consider $h=2$, $m=298$, $n=305$, and any $d\geq 2$. Then
            \begin{align*}
                V_1(m,n)&=\{48,122,298,299,301,302,303,304\},\\
                V_2(m,n)&=\{108,183,298,299,301,302,303,304\}.
            \end{align*}
            As the figure shows, $48\in V_1(m,n)$ because it is a leaf of the tree rooted at $300\in[m..n]$, reached by taking the first child at every level. Similarly, $108\in V_2(m,n)$ because it is a leaf of the tree rooted at 300, reached by taking the second child at every level. When $d$ is less than the depth of the tree, the tree is truncated to depth $d$ before considering the leaves. Thus, the same example with $d=1$ would replace 48 with 120 in $V_1(m,n)$ and 108 with 180 in $V_2(m,n)$.
        }\label{fig:8treeargument}
    \end{figure}
    
    We also define the \emph{closure} of $V_i$ as
    \[\overline{V}_i(m,n)=\left\{\ell\cdot \frac{p_{h+1}^{a_1+\ldots+a_h}}{p_1^{a_1}\cdots p_h^{a_h}}\;\middle|\; \ell\in V_i\text{ and }a_1,\ldots,a_h\in\N\right\}.\]
    An index $k$ belongs to this closure if there is some index $\ell\in V_i(m,n)$ such that $Y[\ell]$ appears \emph{somewhere} in the tree of parity operations used to compute $Y[k]$. This includes the root, so $V_i(m,n)\subseteq \overline{V}_i(m,n)$.

    The following Lemma formalizes the intuition that an $h$-FSG is missing some key piece of information. 

    \begin{lemma}\label{lem:hier-disjoint}
        There is a constant $\gamma\in(1/2,1)\cap \Q$ such that whenever $n$ is sufficiently large and $m\in[\gamma n,n]$, there is some $j(m,n)\in\{1,\ldots,h\}$ satisfying
        \[U(m,n)\cap \overline{V}_{j(m,n)}(m,n)=\emptyset.\]
    \end{lemma}
    \begin{proof}
        For each $i\in\{1,\ldots,h\}$, let
        \[T_i = \left\{ \left( \frac{p_i}{p_{h+1}}\right)^t \frac{p_{h+1}^{a_1+\ldots+a_h}}{p_1^{a_1}\cdots p_h^{a_h}}\in(0,1)\;\middle|\; t\leq d \text{ and } a_1,\ldots,a_h\in\N \right\},\]
        noting that $T_i$ and $T_{i'}$ are disjoint whenever $i\neq i'$. Let
        \[T= \bigcup_{i=1}^{h} T_i.\]
        Since $T$ is finite, there is some $\zeta\in(0,1)$ such that, for all distinct $\tau, \tau' \in T$, we have $|\tau-\tau'|>\zeta$. Choose any rational
        \[\gamma\in\left(\frac{2}{2+\zeta},1\right).\]

        Now let $m,n\in\N$ with $m\in[\gamma n,n]$. Observe that for each $i$,
        \[\overline{V}_i(m,n)=\{\tau k\mid\tau\in T_i\text{ and }k\in[m..n]\}.\]
        Therefore, for any $\ell\in \overline{V}_i(m,n)$ and $\ell'\in\overline{V}_{i'}(m,n)$ with $i\neq i'$, we have $\ell=\tau k$ and $\ell'=\tau'k'$, where $\tau,\tau'\in T$ are distinct and $k,k'\in[m..n]$. Then
        \begin{align*}
            |\ell-\ell'|&=|\tau k-\tau'k'|\\
            &\geq |k\cdot (\tau-\tau')|-|\tau'\cdot(k-k')|\\
            &\geq \gamma n\zeta-(1-\gamma)n.
        \end{align*}
        
        The set $U(m,n)$ is a union of intervals, each of length at most $n-m+1\leq (1-\gamma)n+1$. By our choice of $\gamma$, we have
        \[\gamma n\zeta-(1-\gamma)n> (1-\gamma)n+1\]
        for all sufficiently large $n$, so each interval can have nonempty intersection with at most one of the $\overline{V}_i(m,n)$. Therefore, as there are at most $h-1$ such intervals, there is some $j(m,n)\in\{1,\ldots,h\}$ such that $U(m,n)\cap \overline{V}_{j(m,n)}(m,n)=\emptyset$.
    \end{proof}

    Due to this missing information, the string $Y[[m,n]]$ is nearly incompressible given information from the trailing heads, in the sense of the following lemma.

    \begin{lemma}\label{lem:hier-complex}
        If $n$ is sufficiently large and $m\in[\gamma n,\sqrt{\gamma}n]$, then
        \[K(Y[[m,n]]\mid Y[U(m,n)])\ge (n-m)(1-p_{h+1}^{-d-2}).\]
    \end{lemma} 
    \begin{proof}
        Let $j=j(m,n)$ be as in Lemma~\ref{lem:hier-disjoint}, and write $V_j$ and $\overline{V}_j$ for $V_j(m,n)$ and $\overline{V}_j(m,n)$. Define the set $W=\{0,\ldots,m-1\}\setminus\overline{V}_j$ and the strings $u=Y[U(m,n)]$, $u_h=Y[[m,n]]$, $v_j=Y[V_j]$, and $w=Y[W]$.

        Lemma~\ref{lem:hier-disjoint} implies $U\subseteq W$, so $w\mapsto u$ is computable. It follows that
        \begin{align*}
            K(u_h\mid u)&\geq K(u_h\mid w)-O(1)\\
            &\geq |V_j|-O(\log n)\\
            &\geq (n-m)\left(1-p_{h+1}^{-d-1}\right)-O(\log n),
        \end{align*}
        where the last two inequalities are shown in~\cite{mhfsd}, in the proof of Lemma 5.6; briefly, they hold because $v_j$ can be calculated from $u_h$ given $w$, and $v_j$ contains nearly $(n-m)$ bits that are independent of $w$. Unlike the definition of $u$, the definitions of $u_h$, $w$, and $V_j$ do not depend on the behavior of the trailing heads, so these inequalities are properties of the sequence itself and are unaffected by moving to the adaptive setting.
        
        By decreasing the exponent on $p_{h+1}$ to $-d-2$, the $O(\log n)$ term is absorbed for all sufficiently large $n$ since $n-m\geq n-\sqrt{\gamma}n$.
    \end{proof}
    
    Using this Kolmogorov complexity bound, we now bound the value of the gambler's $s$-gale for $s$ arbitrarily close to 1.
    \begin{lemma}\label{lem:hier-dbound}
        For all $\varepsilon \in \Q \cap (0,1)$, the set
        $\big\{d_G^{(1-\varepsilon/\gamma)}(Y[0..n-1])\;\big|\; n \in \N \big\}$
        is bounded by a constant.
    \end{lemma}

    \begin{proof}
        Let $d$ be sufficiently large so that $p_{h+1}^{-d-2}<\varepsilon/2$, and let $n_0$ be sufficiently large so that Lemma~\ref{lem:hier-complex} applies for all $n\geq n_0$ and Lemma~\ref{lem:dtok} applies for all $n\geq n_0$ and $m\leq\gamma n$. For each $k\geq 1$, let $n_k=\lfloor n_0/\gamma^k\rfloor$. For all $k\in\N$, let $y_k=Y[0..n_k-1]$ and $x_k=Y[n_k..n_{k+1}-1]$.

        For each $k\in\N$, Lemma~\ref{lem:hier-complex} and our choice of $d$ imply
        \[K(Y[[n_k..n_{k+1}-1]]\mid Y[U(n_k,n_{k+1}-1)])\geq (n_{k+1}-1-n_k)(1-\varepsilon/2),\]
        so Lemma~\ref{lem:dtok}, with $\alpha=\varepsilon$ and $\delta=\varepsilon/2$, gives
        \begin{equation}\label{eq:lem:hier-dbound}
            \max_{w\sqsubseteq x_k}\frac{d_G(y_kw)}{d_G(y_k)}\leq 2^{\varepsilon\cdot (n_{k+1}-1-n_k)}\leq 2^{\varepsilon\cdot (n_{k+1}-n_k)}.
        \end{equation}
        Therefore, for each $\ell\in\N$ and $n\in (n_\ell,n_{\ell+1}]$, letting $w=Y[n_\ell..n-1]$, we have
        \begin{align*}
            d^{(1-\varepsilon/\gamma)}_G(Y[0..n-1])&=2^{-\varepsilon n/\gamma}d_G(Y[0..n-1])\\
            &=2^{-\varepsilon n/\gamma}d_G(y_0)\left(\prod_{k=0}^{\ell-1}\frac{d_G(y_kx_k)}{d_G(y_k)}\right)\frac{d_G(y_\ell w)}{d_G(y_\ell)}\\
            &\leq 2^{-\varepsilon n/\gamma}d_G(y_0)2^{\varepsilon n_{\ell+1}}\\
            &\leq 2^{-\varepsilon n/\gamma}d_G(y_0)2^{\varepsilon n/\gamma}\\
            &=d_G(y_0),
        \end{align*}
        which is a constant.
    \end{proof}

    Since Lemma~\ref{lem:hier-dbound} holds for every adaptive $h$-head finite-state gambler and for arbitrarily small $\varepsilon$, we conclude that $\adimFS{(h)}(Y)=1$, completing the proof of Theorem~\ref{thm:hierarchy}.

\section{Conclusion}\label{sec:conc}
    By establishing a separation between the predictive power of oblivious and adaptive multi-head finite-state gamblers, we have shown that adaptive multi-head finite-state dimensions and predimensions are independently interesting objects of study. The statement and proof of our hierarchy theorem demonstrate that this separation is subtle and that some of the pre-existing techniques for analyzing oblivious multi-head finite-state predimensions can be modified to apply in the adaptive setting. This opens several natural directions for further work, some of which we now highlight. First, what is the maximum separation? The separation given by Theorems~\ref{thm:aub} and~\ref{thm:olb} is quantitatively quite small.  Second, are adaptive multi-head finite-state predimensions stable under unions? The technique used by Huang, Li, Lutz, and Lutz~\cite{mhfsd} to prove instability in the oblivious setting relies heavily on obliviousness. Third, does the compression characterization of Lutz~\cite{mhfsc} hold in the adaptive setting? We conjecture that it does. Finally, and more broadly, what would be the effects of allowing two-way head movement or non-determinism? These modifications have been studied in the context of multi-head automata for language recognition~\cite{HKM09}, but not for multi-head gamblers.

\bibliographystyle{splncs04}
\bibliography{amhfsg}

\begin{thebibliography}{10}
\providecommand{\url}[1]{\texttt{#1}}
\providecommand{\urlprefix}{URL }
\providecommand{\doi}[1]{https://doi.org/#1}

\bibitem{AHLM07}
Athreya, K.B., Hitchcock, J.M., Lutz, J.H., Mayordomo, E.: Effective strong dimension in algorithmic information and computational complexity. SIAM Journal on Computing  \textbf{37}(3),  671--705 (2007). \doi{10.1137/S0097539703446912}

\bibitem{BCH2018}
Becher, V., Carton, O., Heiber, P.A.: Finite-state independence. Theory of Computing Systems  \textbf{62},  1555--1572 (2018). \doi{10.1007/s00224-017-9821-6}

\bibitem{DLLM04}
Dai, J.J., Lathrop, J.I., Lutz, J.H., Mayordomo, E.: Finite-state dimension. Theoretical Computer Science  \textbf{310}(1),  1--33 (2004). \doi{10.1016/S0304-3975(03)00244-5}

\bibitem{DowHir10}
Downey, R.G., Hirschfeldt, D.R.: Algorithmic Randomness and Complexity. Springer-Verlag, New York (2010). \doi{10.1007/978-0-387-68441-3}

\bibitem{DKP20}
{\v{D}}uri{\v{s}}, P., Kr{\'a}lovi{\v{c}}, R., Pardubsk{\'a}, D.: Tight hierarchy of data-independent multi-head automata. Journal of Computer and System Sciences  \textbf{114},  126--136 (2020). \doi{10.1016/j.jcss.2020.06.005}

\bibitem{HKM09}
Holzer, M., Kutrib, M., Malcher, A.: Multi-head finite automata: Characterizations, concepts and open problems. In: Electronic Proceedings in Theoretical Computer Science (EPTCS). vol.~1, pp. 93--107 (2009). \doi{10.4204/EPTCS.1.9}

\bibitem{mhfsd}
Huang, X., Li, X., Lutz, J.H., Lutz, N.: Multi-head finite-state dimension (2026), \url{https://arxiv.org/abs/2509.22912}

\bibitem{KozShe2021}
Kozachinskiy, A., Shen, A.: Automatic {K}olmogorov complexity, normality, and finite-state dimension revisited. Journal of Computer and System Sciences  \textbf{118},  75--107 (2021). \doi{10.1016/j.jcss.2020.12.003}

\bibitem{LiVit19}
Li, M., Vit\'{a}nyi, P.M.B.: An Introduction to {K}olmogorov Complexity and Its Applications. Springer-Verlag, Berlin, fourth edn. (2019). \doi{10.1007/978-3-030-11298-1}

\bibitem{DCC}
Lutz, J.H.: Dimension in complexity classes. SIAM Journal on Computing  \textbf{32}(5),  1236--1259 (2003). \doi{10.1137/S0097539701417723}

\bibitem{DISS}
Lutz, J.H.: The dimensions of individual strings and sequences. Information and Computation  \textbf{187}(1),  49--79 (2003). \doi{10.1016/S0890-5401(03)00187-1}

\bibitem{mhfsc}
Lutz, N.: Multi-head finite-state compression. In: LATIN 2026: Theoretical Informatics (to appear)

\bibitem{Mart66}
Martin-L\"of, P.: The definition of random sequences. Information and Control  \textbf{9}(6),  602--619 (1966). \doi{S0019-9958(66)80018-9}

\bibitem{YR78}
Yao, A.C., Rivest, R.L.: $k+1$ heads are better than $k$. Journal of the ACM  \textbf{25}(2),  337--340 (1978). \doi{10.1145/322063.322076}

\end{thebibliography}

\end{document}